%% file: chhk22.tex
\documentclass[12pt]{article}

\input{preamble}

\graphicspath{./}

\begin{document}

\title{\bf Unbiased Statistical Estimation and Valid Confidence Intervals Under Differential Privacy \thanks{Authors XH, JH, GK listed in alphabetical order}}
\author{Christian Covington \thanks{Harvard University. {\tt ccovington@g.harvard.edu}. Work conducted under support from an NSERC Discovery Grant while a student at the University of Waterloo.} \and
      Xi He \thanks{Cheriton School of Computer Science, University of Waterloo. {\tt xihe@uwaterloo.ca}. Supported by an NSERC Discovery Grant. } \and
      James Honaker \thanks{Facebook \& Harvard John A. Paulson School of Engineering and Applied Sciences {\tt james@hona.kr} } \and  
      Gautam Kamath \thanks{Cheriton School of Computer Science, University of Waterloo. {\tt g@csail.mit.edu}. Supported by an NSERC Discovery Grant and a University of Waterloo startup grant.}
      }
\maketitle

\bigskip
\begin{abstract}
  We present a method for producing unbiased parameter 
estimates and valid confidence regions/intervals under the constraints of 
differential privacy, a formal framework for limiting individual information leakage from sensitive data. 
Prior work in this area is limited in that it is
tailored to calculating confidence intervals
for specific statistical procedures, such as mean estimation or simple linear regression.
While other recent work can produce confidence intervals for more general sets of procedures, they 
either yield only approximately unbiased estimates, 
are designed for one-dimensional outputs, or assume significant user knowledge about the data-generating 
distribution. 
Our method induces distributions of mean and covariance estimates via the bag of little bootstraps (BLB)~\citep{KTSJ14} and uses 
them to privately estimate the parameters' sampling distribution 
via a generalized version of the CoinPress estimation algorithm~\citep{BDKU20}. 
If the user can bound the parameters of the BLB-induced parameters and provide heavier-tailed families,
the algorithm produces unbiased parameter estimates and valid confidence intervals
which hold with arbitrarily high probability. These results hold in high dimensions and for any estimation procedure which behaves 
nicely under the bootstrap. 
\end{abstract}

\newpage

\input{introduction}
\input{algorithm_overview}
\input{empirical_evaluation}
\input{discussion}

\bibliographystyle{apalike}
\bibliography{chhk}

\newpage
\appendix
\appendixpage
\input{supplement}

\end{document}

%% file: preamble.tex
\usepackage{ifthen}
\usepackage{mdwlist}
\usepackage{amsmath,amssymb,amsfonts,amsthm,bbm}
\usepackage{bm}
\usepackage[colorlinks=true,linkcolor=red,citecolor=blue,urlcolor=blue]{hyperref}
\usepackage{enumitem}
\usepackage{graphicx}
\usepackage{xspace}
\usepackage{verbatim}
\usepackage{algorithm}
\usepackage[noend]{algpseudocode}
\usepackage[margin=0.8in]{geometry}
\usepackage{color}
\usepackage{thm-restate}
\usepackage{latexsym}
\usepackage{epsfig}
\usepackage{subcaption}
\usepackage{tabularx}
\usepackage{mathrsfs}
\usepackage{xcolor}
\usepackage{tcolorbox}
\usepackage{appendix}
\usepackage{titlesec}
\usepackage[colorinlistoftodos,textsize=scriptsize]{todonotes}
\usepackage[normalem]{ulem}
\usepackage{natbib}
\usepackage{booktabs}
\definecolor{bgcolor}{RGB}{254, 252, 232}
\definecolor{OliveGreen}{rgb}{0,0.6,0}
\def\withcolors{1}
\def\withnotes{1}

\titlespacing\section{0pt}{12pt plus 4pt minus 2pt}{0pt plus 2pt minus 2pt}
\titlespacing\subsection{0pt}{12pt plus 4pt minus 2pt}{0pt plus 2pt minus 2pt}
\titlespacing\subsubsection{0pt}{12pt plus 4pt minus 2pt}{0pt plus 2pt minus 2pt}
\titlespacing\subsubsubsection{0pt}{12pt plus 4pt minus 2pt}{0pt plus 2pt minus 2pt}

\newtheorem{theorem}{Theorem}

\newtheorem{fact}[theorem]{Fact}
\newtheorem{lemma}[theorem]{Lemma}

\newtheorem{corollary}[theorem]{Corollary}

\newtheorem{definition}[theorem]{Definition}

\newtheorem{assumption}[theorem]{Assumption}

\numberwithin{theorem}{section} 
\numberwithin{nontheorem}{section} 
\numberwithin{proposition}{section} 
\numberwithin{observation}{section} 
\numberwithin{remark}{section} 
\numberwithin{fact}{section} 
\numberwithin{lemma}{section} 
\numberwithin{claim}{section} 
\numberwithin{corollary}{section} 
\numberwithin{case}{section} 
\numberwithin{dfn}{section} 
\numberwithin{definition}{section} 
\numberwithin{question}{section} 
\numberwithin{openquestion}{section} 
\numberwithin{res}{section} 
\numberwithin{assumption}{section}
\numberwithin{example}{section} 
\numberwithin{conjecture}{section}

\ifnum\withcolors=1
   
\else

\fi

\ifnum\withnotes=1

\else
  \newcommand{\gnote}[1]{}
  \newcommand{\gfootnote}[1]{}

\fi
\newcommand{\ignore}[1]{\leavevmode\unskip} 

\newcommand{\norm}[1]{\lVert#1\rVert}

\newcommand{\Esymb}{\mathbb{E}}
\newcommand{\Psymb}{\mathbb{P}}
\newcommand{\Vsymb}{\mathbb{V}}
\DeclareMathOperator*{\E}{\Esymb}
\DeclareMathOperator*{\Var}{\Vsymb}
\DeclareMathOperator*{\ProbOp}{\Psymb}
\renewcommand{\Pr}{\ProbOp}

\newcommand\bdot\bullet

\newcommand{\R}{\mathbb R}

\renewcommand{\leq}{\leqslant}

\renewcommand{\geq}{\geqslant}

\let\epsilon=\varepsilon
\newcommand\MYcurrentlabel{xxx}
\newcommand{\MYstore}[2]{%
	\global\expandafter \def \csname MYMEMORY #1 \endcsname{#2}%
}
\newcommand{\MYload}[1]{%
	\csname MYMEMORY #1 \endcsname%
}
\newcommand{\MYnewlabel}[1]{%
	\renewcommand\MYcurrentlabel{#1}%
	\MYoldlabel{#1}%
}
\newcommand{\MYdummylabel}[1]{}
\newcommand{\torestate}[1]{%
	\let\MYoldlabel\label%
	\let\label\MYnewlabel%
	#1%
	\MYstore{\MYcurrentlabel}{#1}%
	\let\label\MYoldlabel%
}
\newcommand{\restatetheorem}[1]{%
	\let\MYoldlabel\label
	\let\label\MYdummylabel
	\begin{theorem*}[Restatement of \cref{#1}]
		\MYload{#1}
	\end{theorem*}
	\let\label\MYoldlabel
}
\newcommand{\restatelemma}[1]{%
	\let\MYoldlabel\label
	\let\label\MYdummylabel
	\begin{lemma*}[Restatement of \cref{#1}]
		\MYload{#1}
	\end{lemma*}
	\let\label\MYoldlabel
}
\newcommand{\restateprop}[1]{%
	\let\MYoldlabel\label
	\let\label\MYdummylabel
	\begin{proposition*}[Restatement of \cref{#1}]
		\MYload{#1}
	\end{proposition*}
	\let\label\MYoldlabel
}
\newcommand{\restatefact}[1]{%
	\let\MYoldlabel\label
	\let\label\MYdummylabel
	\begin{fact*}[Restatement of \cref{#1}]
		\MYload{#1}
	\end{fact*}
	\let\label\MYoldlabel
}
\newcommand{\restate}[1]{%
	\let\MYoldlabel\label
	\let\label\MYdummylabel
	\MYload{#1}
	\let\label\MYoldlabel
}

\allowdisplaybreaks
\DeclareMathOperator{\diag}{diag}

\makeatletter
\newcommand{\distas}[1]{\mathbin{\overset{#1}{\kern\z@\sim}}}%
\newsavebox{\mybox}\newsavebox{\mysim}
\newcommand{\distras}[1]{%
	\savebox{\mybox}{\hbox{\kern3pt$\scriptstyle#1$\kern3pt}}%
	\savebox{\mysim}{\hbox{$\sim$}}%
	\mathbin{\overset{#1}{\kern\z@\resizebox{\wd\mybox}{\ht\mysim}{$\sim$}}}%
}
\makeatother

\def\*#1{\mathbf{#1}}

\def\HiLi{\leavevmode\rlap{\hbox to \hsize{\color{yellow!50}\leaders\hrule height 1.2 \baselineskip depth .5ex\hfill}}}

\newcommand{\mvmeaniter}{\textsc{MVMRec}}
\newcommand{\mvmeanstep}{\textsc{MVM}}

\newcommand{\gvdp}{\textsc{GVDP}}
\newcommand{\blb}{\textsc{BLB}}

\definecolor{dred}{rgb}{0.75, 0.0, 0.2}
\definecolor{dgreen}{rgb}{0.33, 0.42, 0.18}
\definecolor{dblue}{rgb}{0.0, 0.2, 0.4}
\definecolor{crimson}{rgb}{0.86, 0.08, 0.24}

%% file: introduction.tex
\section{Introduction}
\label{sec:introduction} 

\subsection{Overview}
\label{subsection:overview}

The dramatic expansion of data collection and analysis has led to growing concerns around the role of privacy and security in the modern world.
Our particular focus will be on statistical analysis and how it can leak information about individuals 
in a data set being analyzed. 

Statistical agencies, in particular, have been concerned with 
\emph{statistical disclosure limitation}, a broad set of techniques used to limit the leakage of sensitive information from statistical analyses.
\citet{DL86, DL89} and \citet{Rei05} 
work on identifying and quantifying disclosure risk under different assumptions. There has also been substantial work developing various privacy definitions such as 
k-anonymity~\citep{Swe02}, t-closeness~\citep{LTV07}, and l-diversity~\citep{MKGV07}
among others.

\citet{DN03} developed a polynomial data reconstruction algorithm and used it to prove a result 
later coined in~\citet{DR14} as the \emph{Fundamental Law of Information Recovery}.
Roughly, this law states that an attacker can reconstruct a data set by asking a sufficiently large 
number of cleverly chosen queries of the data.
This inspired the invention of a more formal privacy definition called \emph{differential privacy} (DP)~\citep{DMNS06}.
DP is a definition which requires that, for any two data sets differing in one row, 
the change in the distribution of answers to any possible query between the two data sets is minimal. 
We loosely define a ``DP Algorithm'' as a procedure that takes a statistical estimator and 
converts it into an estimator that satisfies DP (i.e.\ a ``DP Estimator''); we define what it means to ``satisfy DP'' later. 

DP has become a popular tool in many corners of industry but has not been widely applied to research
in many fields that often analyze sensitive data (social sciences, medicine, etc.).
We suggest that this is, in part, because of a lack of DP algorithms that effectively meet the 
needs of these fields. First, DP algorithms typically require the user to, 
\emph{without looking at the data}, specify a domain to which the data will be censored/clipped.
This is because privacy mechanisms often take the form of additive noise, with a variance parameter 
that scales with the sensitivity of the function being privatized 
(see Definition~\ref*{appendix:defn:global_sensitivity} and Lemma~\ref*{appendix:lemma:gaussian_mechanism}
for more details, and note that
any theorem, algorithm, section, etc.~whose number is preceded by an $S$ can be found in the online supplement).
Many functions of interest, e.g.\ means and variances,
are unbounded (and thus have infinite sensitivity) when defined over unbounded data domains, so we achieve 
finite sensitivity by bounding the data domain, i.e.\ clipping/censoring the data. 

We claim that bounding the data domain effectively is a difficult problem in general, and potentially 
introduces substantial error into the DP pipeline which is difficult to account for.
In particular, without assumptions that the data bounds are set well, DP estimators do not typically yield basic statistical guarantees 
that many applied researchers desire, namely unbiasedness and valid confidence intervals. 
Our goal is to provide a framework (i.e. a DP algorithm) for converting non-private estimators 
to DP estimators in a way that jointly addresses both of these concerns.

\subsection{Problem Demonstration}
\label{subsection:problem_demonstration}

Many DP mechanisms take the form ``non-private statistic plus zero-mean noise'' 
(e.g.\ the \emph{Gaussian mechanism} defined in Lemma~\ref*{appendix:lemma:gaussian_mechanism}) and thus look 
like they ought to be unbiased and yield easily characterizable confidence regions.
However, this doesn't take into account the necessary step of choosing a bounded data domain.
In practice, analysts are required to specify this domain without looking at the data.
The data are then projected into the domain, and analysis proceeds on the projected data.
This introduces a potential trade-off in the analyst's decision calculus.
The analyst generally wants to specify a small domain, 
as this generally requires less noise addition to satisfy DP. 
However, if the analyst's domain is too small, they risk clipping the data, which can lead to 
biased estimates.

Consider the case where $X = \{X_1, \hdots, X_n\}$ with $X_i \sim N(0,1)$ and
$y = X \beta + \epsilon$ for $\beta = 100, \epsilon \sim N(0, 10^2)$.
We estimate $\beta$ and get associated 95\% confidence intervals using OLS 
and test the effect of various levels of clipping on the estimates and confidence intervals. 
Specifically, we leave $X$ unclipped and clip the top $\{0, 0.1, 1, 5\}$ 
percent of $y$. 

\begin{figure}[H]
    \captionsetup{font=footnotesize,labelfont=footnotesize}
    \begin{subfigure}[t]{0.32\textwidth}
        \centering
        \includegraphics[width=\textwidth]{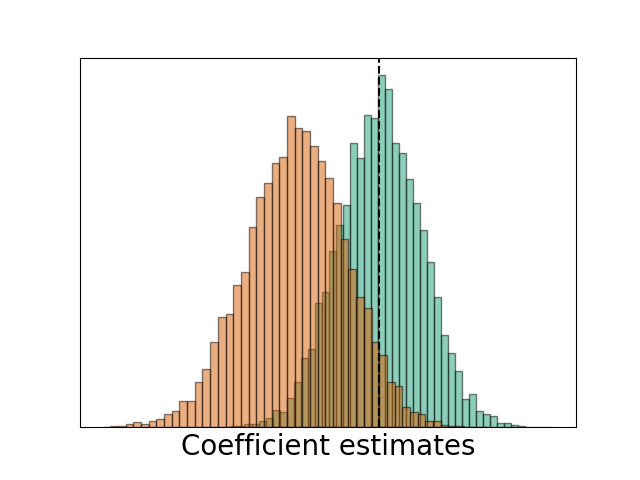}
        \caption{Bias: 0.1\% clipping}
    \end{subfigure}
    \begin{subfigure}[t]{0.32\textwidth}
        \centering
        \includegraphics[width=\textwidth]{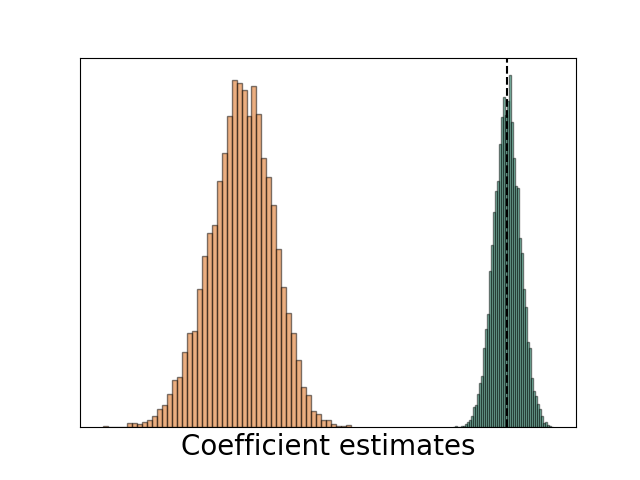}
        \caption{Bias: 1\% clipping}
    \end{subfigure}
    \begin{subfigure}[t]{0.32\textwidth}
        \centering
        \includegraphics[width=\textwidth]{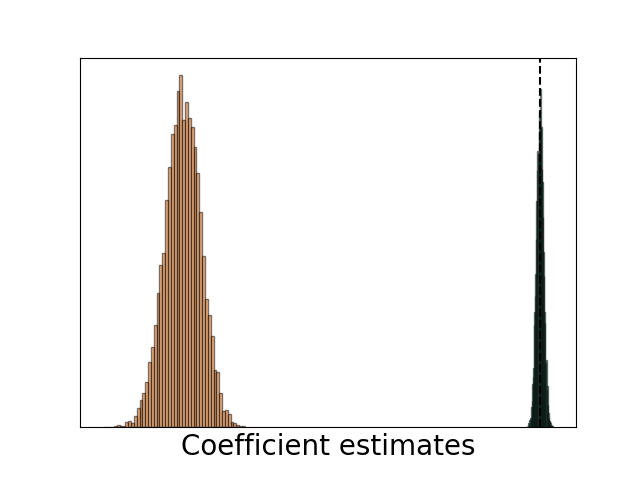}
        \caption{Bias: 5\% clipping}
    \end{subfigure}
    \begin{subfigure}[t]{0.32\textwidth}
        \centering
        \includegraphics[width=\textwidth]{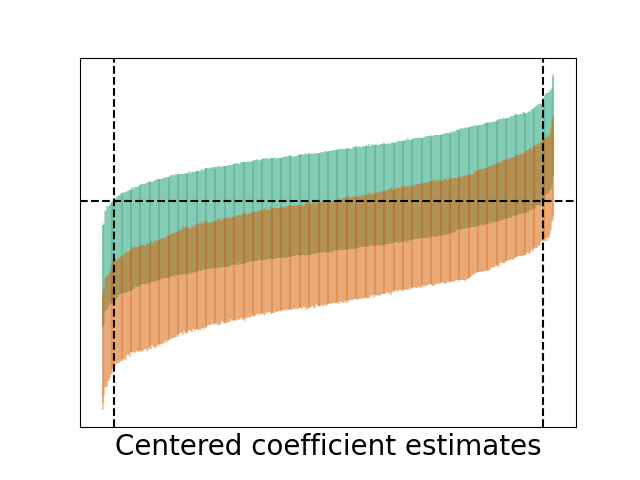}
        \caption{CI coverage: 0.1\% clipping}
    \end{subfigure}
    \begin{subfigure}[t]{0.32\textwidth}
        \centering
        \includegraphics[width=\textwidth]{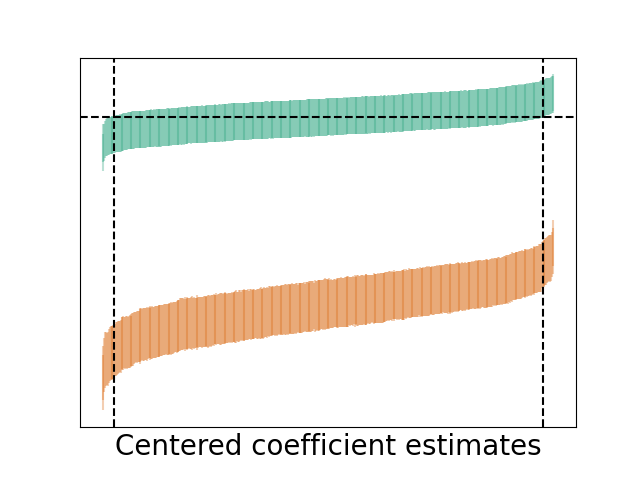}
        \caption{CI coverage: 1\% clipping}
    \end{subfigure}
    \begin{subfigure}[t]{0.32\textwidth}
        \centering
        \includegraphics[width=\textwidth]{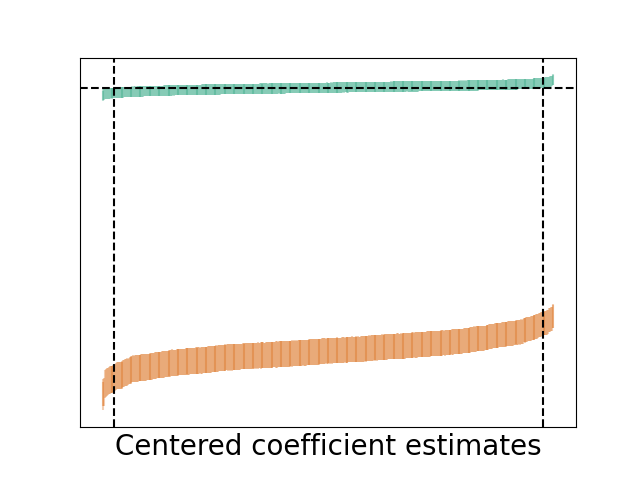}
        \caption{CI coverage: 5\% clipping}
    \end{subfigure}

    \caption{Distribution of OLS coefficient estimates (a-c) and 95\% confidence intervals (d-f) under different levels of clipping of $y$.
             Non-clipped distribution in green, clipped distribution in orange.}
    \label{figure:ols_clipping}
\end{figure} 

Figure~\ref{figure:ols_clipping} shows results from $10{,}000$ simulations of this process. 
The top plots show coefficient distributions under each level of clipping and compare it to the condition of 
no clipping. Note that at even a moderate level of $1\%$ clipping, the distributions of estimates are completely 
non-overlapping. The bottom plots show the absolute error of the estimates, arranged in increasing order on the x-axis, with vertical bars 
representing the $95\%$ confidence interval for that estimate. The black dotted vertical lines on the left and right sides of each plot show the $0.025$ and $0.975$
quantiles, where we expect perfectly calibrated confidence intervals to cross the x-axis. At $0.1\%$ 
clipping, approximately half of our confidence intervals do not contain the true parameter value;
at higher levels of clipping, none of them do.

We chose a very simple regime for the experiments above: 
one-dimensional OLS with a Gaussian covariate, Gaussian error, clipping in only 
the outcome variable, and no attempt to satisfy DP (i.e.\ no noise addition). 
In more complex settings, the effect of clipping on the coefficients could be larger, 
in addition to being harder to predict and reason about.
Thus, we argue that the data bounding step that precedes so many DP algorithms has 
potentially immense consequences for practical analysis and ought to be seriously considered.

\subsection{Related Work}
\label{subsection:related_work}

Differential privacy has grown in popularity in recent years, as has the literature
exploring the intersection between statistics and DP. \citet{DL09} point out how a handful of common 
robust statistical estimators could be extended to satisfy differential privacy. 
\citet{WZ10} compare DP mechanisms via convergence rates of distributions and densities 
from DP releases and frame DP in statistical language more broadly. \citet{LCSSF16} explores 
model selection under DP, while \citet{PB21} explores model uncertainty. 
\citet{VS09, WLK15, GLRV16, CKMSU19}, and \citet{AS19}
propose methods for DP hypothesis testing in various domains. 
\citet{She17, Wan18, BRMY19}, and~\citet{AMSSV20} all address the problem of 
differentially private linear regression.

\citet{KV18} gives nearly optimal confidence intervals for univariate Gaussian mean estimation with finite sample 
guarantees. \citet{DFM+20} proposes their own algorithms for the same problem and finds superior practical 
performance in some domains, and \citet{BDKU20} develop an algorithm that works well at reasonably small sample sizes
and without strong assumptions on user knowledge, 
while also scaling well to high dimensions. 
\citet{DGH+21} explores non-parametric confidence intervals for calculating medians.
\citet{DHK15} gives confidence intervals for a difference of means.
\citet{ABL21} shows how to construct DP version of M-estimators, as 
well as associated confidence regions.

Our work continues in a line of recent work for constructing confidence intervals for more general classes of 
differentially private estimators.
\citet{BH18} shows how to combine estimates from additive functions that satisfy zCDP to get confidence intervals 
at no additional cost. 
\citet{WKL19} provides confidence intervals for models trained with objective or 
output perturbation algorithms. These algorithms are quite general, but require 
solving the non-private ERM sub-problem optimally. 
\citet{FWS21} presents a very general approach based on privately estimating parameters of the 
data-generating distribution and bootstrapping confidence intervals by 
repeatedly running the model of interest on samples from a distribution parameterized by the privately estimated parameters.
This method is efficient with respect to its use of the privacy budget, but relies on significant knowledge of 
the structure of the data-generating process.
With the exception of~\citet{KV18}, these works either ignore the issue of bounding the data domain $\mathcal{X}$ effectively, 
or attempt to address it through bias-correction strategies which we believe are unlikely to work for complex problems.
\citet{BWSB21} tests a variety of DP algorithms for various tasks (ranging from 
tabular statistics to OLS regression) on real data and argue that existing methods for 
performing DP regression, ``would struggle to produce accurate regression estimates and confidence intervals'' 
(\citet{BWSB21}, p. 1).

\citet{EKST19} is the closest existing work to ours. They also start with the Sample \& Aggregate framework,
and BLB algorithm to get $k$ estimates of the parameters of the sampling distribution of their non-private estimator.
The $k$ estimates are then aggregated via a differentially private mean and confidence intervals are calculated 
using a differentially private variance estimate and CLT assumption. Because the estimates are projected into a 
bounded paramter domain to control the sensitivity of the mean, the resulting private mean is not necessarily unbiased. \citet{EKST19} attempts to address this issue by privately estimating the proportion of the $k$ estimates 
that are clipped by the projection and adjusting the private mean by the estimated clipping proportions.
This method has the advantage of allowing users to specify overly tight clipping bounds in order 
to decrease the global sensitivity of their estimator, but 
is sensitive to how well the clipping proportions are estimated and, to our knowledge,
has no means of generalizing to multivariate parameter estimation.

\subsection{Contributions}
\label{subsection:contributions}

We introduce a general-purpose meta-algorithm that allows an analyst to take 
any estimator that is (1) unbiased 
in the non-private setting and (2) has ``nice'' properties under the bootstrap and
produce a version that satisfies DP and, 
with high probability, is unbiased and produces valid confidence intervals or a 
valid confidence region.
Our results hold under the central (or trusted curator) model of DP.

Our algorithm can be split into three distinct steps, each of which we explain in detail in Section~\ref{section:algorithm_overview}.
First, we use the Bag of Little Bootstraps (BLB) algorithm (Algorithm~\ref*{appendix:alg:blb}) developed by~\citet{KTSJ14}
to produce estimates of the mean and covariance of the sampling distribution of the non-private estimator. 

Second, we privately estimate the mean of each set of BLB estimates (both the means and covariances) using a modified 
version of the CoinPress private mean estimation algorithm (Algorithm~\ref*{appendix:alg:hd_mean_est} \citep{BDKU20}).
For both the mean and covariance distributions induced by the BLB, 
the analyst must provide a distribution that is heavier-tailed 
(Definition~\ref*{appendix:defn:heavier_tails} and Assumption~\ref{assump:user_heavier_tailed}),
as well as give bounds on the mean (Assumption~\ref{assump:initial_parameter_bounds})
and covariance (Assumption~\ref{assump:initial_covariance_bounds}) of the induced distribution.
Under these conditions, this step produces parameter estimates which are unbiased and follow a
multivariate Gaussian distribution
with known covariance
(Theorem~\ref{theorem:mean_est_form}).
We argue that these conditions are natural and demonstrate how the 
properties of the CoinPress algorithm allow the analyst to get good performance 
even when they set the aforementioned bounds very conservatively. 

Although our use of CoinPress in this way guarantees that DP is satisfied with respect to the BLB estimates, 
this guarantee also holds for the underlying sensitive data over which the BLB 
estimates were calculated. This fact follows from noting that combining the BLB with 
an aggregation step that satisfies DP (such as CoinPress) falls under the purview of the 
Sample \& Aggregate framework developed in~\citet{NRS07}. 

Third, we combine the estimated parameters using precision weighting to produce
final mean and covariance estimates 
(Theorems~\ref{thm:tilde_theta} and~\ref{theorem:tilde_Sigma}),
which are used to calculate a valid confidence region/intervals 
(Theorem~\ref{theorem:general_cr_high_prob}).
We do so via a multivariate extension of the precision-weighting technique to improve CoinPress' estimates 
(Theorem~\ref{theorem:precision_weighting}). 
While precision-weighting is a well-known technique in the meta-analysis literature~\cite{Coc54}, 
we give, to the best of our knowledge, the first proof of multivariate optimality, 
which may be of independent interest.
This step maintains the privacy guarantees of the CoinPress algorithm 
because differential privacy is preserved under postprocessing (Lemma~\ref{lemma:zcdp_postprocessing}).

We believe that our framework is a promising step toward making differential privacy more practical 
for applied research. 
The problem of choosing good data bounds 
is significant in practice in that it is both generally difficult and that
many DP algorithms are sensitive to poor choices. There is also currently an asymmetry in the 
failure modes, in that bounds that are too wide typically yield 
answers which are unbiased but very noisy, while bounds that are too narrow
risk ``silent failure'', where the DP result looks precise but is not actually 
representative of the non-private answer.  
Our framework, through use of the CoinPress algorithm, 
gives users more leeway to err on the side of conservatism without introducing bias, thus mitigating the 
possibility of getting DP results that appear precise but are systematically incorrect.

Moreover, our framework is general enough to be applied to any estimator
for which the BLB does a ``good'' job approximating the sampling distribution
of the estimator. 
The bootstrap is broadly familiar to applied statisticians, and its properties 
for any particular estimator an analyst wishes to use can, in principle, be tested on non-sensitive data. 
Thus, answering the question of whether or not our algorithm will be useful in a particular setting
does not require significant knowledge of DP.

%% file: algorithm_overview.tex
\section{Algorithm Overview}
\label{section:algorithm_overview}

\subsection{DP Preliminaries}
\label{subsection:dp_preliminaries}

Throughout this work we use a particular notion of DP called zero-concentrated DP.
Suppose we have a data domain $\mathcal{X}$ and data set $X \in \mathcal{X}^n$.

\begin{definition}[Zero-concentrated differential privacy (zCDP)~\citep{BS16}]
    \label{defn:zcdp}
    Let $\mathcal{M}: \mathcal{X}^n \rightarrow \Omega$ be a randomized algorithm where 
    $\left( \Omega, \Sigma, \Pr \right)$ is a probability space and $\rho \geq 0$.
    We say that $\mathcal{M}$ satisfies $\rho$-zCDP with respect to 
    a data set $X$ if, for all 
    $(X', X^*) \in \mathcal{X}_n$ and $\alpha \in (1, \infty)$:
    \[
        H_{\alpha} \left( \mathcal{M}(X') \Vert \mathcal{M}(X^*) \right) \leq \rho \alpha,
    \]
    where $H_{\alpha}$ is the $\alpha$-R{\'e}nyi divergence.
\end{definition}

The parameter $\rho$ represents an upper bound on the amount of information $\mathcal{M}$ leaks 
about the underlying data $X$. Larger $\rho$ implies more information leakage, or \emph{privacy loss},
but also allows for the statistics returned by $\mathcal{M}$ to be more accurate. 

zCDP (like other standard notions of DP) has two properties which are very useful for reasoning about how DP guarantees operate
within a full data analysis pipeline.

\begin{lemma}[Composition of zCDP~\citep{BS16}]
    \label{lemma:zcdp_composition}
    Let $\mathcal{M}: \mathcal{X}^n \rightarrow \mathcal{Y}$ and 
    $\mathcal{M}': \mathcal{X}^n \rightarrow \mathcal{Z}$ such that
    $\mathcal{M}$ satisfies $\rho$-zCDP and $\mathcal{M}'$ satisfies $\rho'$-zCDP.
    Define $\mathcal{M}'' : \mathcal{X}^n \rightarrow \mathcal{Y} \times \mathcal{Z}$ by 
    $\mathcal{M}''(x) = \left( \mathcal{M}(x), \mathcal{M}'(x) \right)$.
    Then $\mathcal{M}''$ satisfies $(\rho + \rho')$-zCDP.
\end{lemma}

We use zCDP because its privacy parameters $\rho$ compose additively, which is convenient for algorithms,
like CoinPress, that contain multiple private releases.
zCDP implies the more familiar notion of $(\epsilon, \delta)$-DP (see Proposition 3 of~\cite{BS16}),
so any zCDP guarantees in this paper can be converted to $(\epsilon,\delta)$-DP if an analyst prefers.

\begin{lemma}[Postprocessing of zCDP~\citep{BS16}]
    \label{lemma:zcdp_postprocessing}
    Let $\mathcal{M}: \mathcal{X}^n \rightarrow \mathcal{Y}$ and 
    $f: \mathcal{Y} \rightarrow \mathcal{Z}$ such that $\mathcal{M}$ 
    satisfies $\rho$-zCDP. Define $\mathcal{M}': \mathcal{X}^n \rightarrow \mathcal{Z}$
    such that $\mathcal{M}'(x) = f \left( \mathcal{M}(x) \right)$.
    Then $\mathcal{M}'$ satisfies $\rho$-zCDP.
\end{lemma}

This postprocessing property of zCDP states that if some output satisfies DP with respect to some data $X$,
functions of that output also satisfy DP with respect to $X$ (provided that the functions do not take $X$ as input).

\subsection{Algorithm Step 1: Bag of Little Bootstraps and Sample \& Aggregate}
\label{subsection:algorithm_step_1}

In Step 1 of our algorithm, the algorithm takes the analyst's non-private estimator 
$\hat{\theta}$ and uses the Bag of Little Bootstraps to try to approximate the sampling 
distribution of $\hat{\theta}$. In particular, the BLB partitions the data into $k$ disjoint subsets,
repeatedly runs the estimator over each subset, and produces $k$ estimates of its mean, $\{\hat{\theta}_i^{BLB}\}_{i \in [k]}$,
and covariance, $\{\hat{\Sigma}_i^{BLB}\}_{i \in [k]}$.

Let $\mathcal{X}$ be our data domain, $\mathcal{D}$ a distribution over the domain, and $X = \{x_1, \hdots, x_n\}$ be our sensitive data set 
where $x_i$ are drawn i.i.d.\ from $\mathcal{D}$. For shorthand, we say that $X \in \mathcal{X}^n$ and $X \sim \mathcal{D}^n$.
We say that the analyst wants to run some model, which has an associated parameter vector $\theta \in \R^d$.
The analyst specifies the estimator they would have liked to run in the non-private setting 
$\hat{\theta}: \mathcal{X}^{n} \rightarrow \R^d$. Our goal is eventually to approximate $\hat{\theta}$ in a manner that 
satisfies DP with respect to $X$.
We assume that $n$ is ``public knowledge'' and does not need to be privately estimated.
If this is not true, we can generate a DP estimate of $n$, call it $\tilde{n}$, and then $X$ 
could be subsampled or augmented with rows of synthetic data 
until it has $\tilde{n}$ rows, creating a new data set $\tilde{X}$ over which we can apply our algorithm. 

As stated earlier, DP algorithms typically require specification of the global sensitivity of the 
function whose outputs are being privatized (see Definition~\ref*{appendix:defn:global_sensitivity} and 
Lemma~\ref*{appendix:lemma:gaussian_mechanism}).
This can become arbitrarily complex for complicated models, even after assuming a bounded input domain. 
The Sample \& Aggregate framework introduced in~\citet{NRS07} provides a strategy for estimating such functions
without specifying the global sensitivity.
First run the function of interest non-privately over $k$ disjoint subsets of the data, bound 
the outputs, and then aggregate the results using a function with a sensitivity that is 
easier to reason about (for our purposes, we assume the aggregation function is the mean). 
We then privately estimate the mean of these $k$ results.
Because each element in the original data 
contributes to one subset of the partition, its effect on the aggregation is localized 
to one of its $k$ inputs, and so a mean estimation algorithm that satisfies DP with respect to those $k$ elements 
also satisfies DP with respect to the underlying data. 
A more thorough treatment of this framework can be found in~\citet{NRS07} and Chapter 7 of~\citet{DR14}.

The Bag of Little Bootstraps algorithm (developed in~\citep{KTSJ14} and reproduced in Section~\ref*{appendix:bag_of_little_bootstraps})
randomly partitions the data $X$ into $k$ disjoint subsets 
$\{X_1, \hdots, X_k\}$, 
scales the subset back up to an effective sample size that matches that of the original data (via multinomial sampling), 
and runs $\hat{\theta}$
$r$ times, producing estimates $\{\hat{\theta}_{i,a}\}_{a \in [r]}$. It then aggregates these 
into an arbitrary assessment of estimator quality. We use the mean and covariance,
so for each $i \in [k]$ we get 
$
    \hat{\theta}^{BLB}_i = \frac{1}{r} \sum_{a=1}^{r} \hat{\theta}_{i,a} \text{ and }
    \hat{\Sigma}^{BLB}_i = \text{Cov} \left( \{\hat{\theta}_{i,a}\}_{a \in [r]} \right).
$
Our approach allows us to find a single confidence region for the entire parameter vector jointly. However, 
we acknowledge that many analysts will prefer separate confidence intervals for each element in their parameter 
vector. This preference is advantageous from a privacy perspective, as confidence intervals will require less noise to privatize than 
a full confidence region and will thus be more accurate.

We present methods and results for the general case of finding a confidence region, but also show how this could be adapted to instead find confidence intervals. We include the setup for finding confidence intervals in the main text with results in the supplement.

These sets of estimates are now empirical approximations to the theoretical distributions of the BLB estimates,
which we assume are themselves good approximations of the actual parameters of interest.
We make this last assumption explicit as follows.

\begin{assumption}
	\label{assump:blb_approx}
    Let the estimator $\hat{\theta} \sim G(\theta, \Sigma)$ where the marginals of $G$ each belong 
    to a location-scale family. 
    For $\hat{\theta}^{BLB}_i$ and $\hat{\Sigma}^{BLB}_i$ generated by applying BLB to our estimator 
    $\hat{\theta}$, let $\hat{\theta}^{BLB} = \frac{1}{k} \sum_{i=1}^{k} \hat{\theta}_i^{BLB}$ and 
    $\hat{\Sigma}^{BLB} = \frac{1}{k} \sum_{i=1}^{k} \hat{\Sigma}_i^{BLB}$.
    We assume that 
	\begin{align*}
        \E \left( \hat{\theta}^{BLB} \right) = \E \left( \hat{\theta} \right) &= \theta \\
        \Pr \left( \hat{\Sigma}^{BLB} \succeq \hat{\Sigma} \right) &= 1.
    \end{align*}
\end{assumption}

We take $\succeq$ to mean ``greater/equal in L{\"o}wner order''. In particular,
$A \succeq B \iff A-B$ is a PSD matrix, in which case we call $A$ a L{\"o}wner upper bound on $B$.

\subsection{Algorithm Step 2: Generating Private Parameter Estimates With CoinPress}
\label{subsection:algorithm_step_2}

Now that we have BLB estimates $\{ \hat{\theta}^{BLB}_i \}_{i \in [k]}$ 
and $\{ \hat{\Sigma}^{BLB}_i \}_{i \in [k]}$, we can privately estimate their empirical means
$\hat{\theta}^{BLB}$ and $\hat{\Sigma}^{BLB}$ using the CoinPress mean estimation algorithm (Section~\ref*{appendix:subsection:modified_coinpress_algorithm}).
We try to provide intuition here for how and why CoinPress works, but interested readers should consult~\citet{BDKU20} for a 
more complete treatment.

For generality, we say that we have 
data $\{ y_i \}_{i \in [k]}$ with empirical mean $\hat{\mu}$, where each $y_i$ is an i.i.d.\ instantiation of a random variable $Y$ 
with mean $\mu_Y$ and covariance $\Sigma_Y$. 
Our goal is to estimate $\mu_Y$ in a manner that satisfies DP.
In order to state these results generally, we will assume that $Y \in \R^{d'}$ for some $d'$.
When $y_i$ stands in for $\hat{\theta}^{BLB}_i$, we have $d' = d$, the dimension of our parameter vector of interest.
The same is true when $y_i$ stands in for $\hat{\Sigma}^{BLB}_i$ and we care only about univariate confidence intervals
such that $\hat{\Sigma}^{BLB}_i = \hat{V}^{BLB}_i I_d$ for some $\hat{V}^{BLB}_i$. When we care about a joint confidence region, we need to use the entire 
upper triangular of $\hat{\Sigma}^{BLB}_i$, so $d' = \frac{d(d+1)}{2}$. 
In the actual implementation (and proofs), these covariance estimates are flattened into a vector of the appropriate 
dimension before estimation and unflattened at the end of estimation
to produce full covariance matrices again. We assume that this is going on behind the scenes and,
for notational clarity, keep calling them $\hat{\Sigma}^{BLB}_i$.

The high-level idea behind CoinPress is 
to make a series of mean estimates, where each estimate (probabilistically) improves upon the last,
making only a few requirements of the analyst. 

\begin{assumption}
    \label{assump:user_heavier_tailed}
    The analyst provides a location-scale family of distributions 
	$Q_Y(\mu, \Sigma)$
	with heavier tails 
    than the distribution of $Y$ as described in Definition~\ref*{appendix:defn:heavier_tails}. 
\end{assumption}

\begin{assumption}
    \label{assump:initial_parameter_bounds}
    The analyst provides 
    $\tilde{\mu}_0 \in \R^d$ and $r_0 \in \R$ such that 
    $\mu_Y \in B_2(\tilde{\mu}_0,r_0)$, the 
    $\ell_2$ ball centered at $\tilde{\mu}_0$ with radius $r_0$. 
\end{assumption}

\begin{assumption}
    \label{assump:initial_covariance_bounds}
    The analyst provides $\Sigma^U_{Y} \in \R^{d \times d}$ such that 
    $\Sigma_Y \preceq \Sigma^U_{Y}$.
\end{assumption}

These three assumptions provide the backbone of the iterative
improvement in CoinPress.
Recall that the noise added to privatize an estimator defined over an input domain generally scales with the size of the domain.
Over a series of $t$ steps, CoinPress attempts to find a small domain that, with high probability, contains all the
$\{ y_i \}_{i \in [k]}$. 
Assumption~\ref{assump:initial_parameter_bounds}
ensures that the algorithm starts with sufficiently conservative bounds that
contain $\mu_X$.
Assumptions~\ref{assump:user_heavier_tailed} and~\ref{assump:initial_covariance_bounds}
then allow the algorithm to convert bounds on $\mu_X$ to high-probability bounds
on the individual data points $y_i$.
For the remainder of the section, we assume that our three assumptions hold.

At step $m \in [t]$ of the algorithm, CoinPress takes the $\ell_2$ ball and expands it outward based on the 
assumed distribution. If you have a distribution of a known family with bounded mean and covariance, you can, with high 
probability, upper bound the $\ell_2$ norm of an arbitrary number of draws from said distribution.
In the context of our case, it specifies a ball that, with high probability, contains all of the $\{ y_i \}_{i \in [k]}$. 
In this case, using the ball as our data domain and projecting our data 
into this domain will not clip any of the data.
The global sensitivity of the mean is calculated using this ball and CoinPress adds noise scaled 
to the sensitivity using the Gaussian mechanism (Lemma~\ref*{appendix:lemma:gaussian_mechanism}), which adds zero-mean noise from a 
multivariate Gaussian with diagonal covariance to get a private estimate $\tilde{\mu}_m$. 
Because we are, with high probability, not clipping any of the 
$\{ y_i \}_{i \in [k]}$, our use of the Gaussian mechanism implies that the form of our 
private estimator is ``empirical mean + zero-mean noise''. 
Thus, based on the scale of the noise, we can produce a new $\ell_2$ ball which, with high probability, contains the true empirical mean.
This ball becomes the bounding set for the mean at the next step and the process continues.
The guarantee that CoinPress does not clip any points ensures (with high probability) 
that for all $m \in [t]$, the $\tilde{\mu}_m$ are unbiased estimates of the empirical mean of the $y_i$.

Moreover, assuming no clipping, the variances (induced by clipping and our use of the Gaussian mechanism) of $\tilde{\mu}_m$ are
just the diagonal of the covariance parameter used in the Gaussian mechanism,
which we call $\vec{\sigma}^2_{\tilde{\mu}, m} \in \R^{d'}_{+}$. 
In other words, under no clipping, the additional error incurred by the general privacy mechanism is 
simply the error from the Gaussian mechanism.

In general, applying the CoinPress algorithm to a data set $\{ y_i \}_{i \in [k]}$
satisfies $\rho$-zCDP with respect to $\{y_i\}_{i \in [k]}$~\citep{BDKU20}.
In our case $\{y_i\}_{i \in [k]}$ stands in for either of the sets of BLB estimates 
$\{\hat{\theta}^{BLB}_{i}\}_{i \in [k]}$ and $\{\hat{\Sigma}^{BLB}_i\}_{i \in [k]}$.
So, CoinPress satisfies zCDP with respect to each set of BLB estimates and, by the extension implied by 
our use of the Sample \& Aggregate framework, also satisfies zCDP with respect to the original data $X$.
Moreover, CoinPress comes with a high-probability guarantee on the form of the private estimates it produces.

\begin{theorem}
	\label{theorem:mean_est_form}
    Under Assumptions~\ref{assump:user_heavier_tailed}, \ref{assump:initial_parameter_bounds}, and \ref{assump:initial_covariance_bounds},
	CoinPress produces $t$ mean estimates $\{ \tilde{\mu}_m \}_{m \in [t]}$ and associated privacy variances 
    $\{ \vec{\sigma}^2_{\tilde{\mu},m} \}_{m \in [t]}$ 
    such that 
    \[ \Pr \left[ \forall m \in [t]: \tilde{\mu}_m \sim N \left( \hat{\mu}, \vec{\sigma}^2_{\tilde{\mu},m} I_{d'} \right) \right] \geq 1 - \beta^{\tilde{\mu}}. \]
\end{theorem}

We now recall that $\hat{\mu}$ is the empirical mean of the $\{y_i\}_{i \in [k]}$, where 
$y_i$ stands in for either $\hat{\theta}^{BLB}_i$ or $\hat{\Sigma}^{BLB}_i$. 
That is, we use CoinPress to privately estimate the mean of the distributions induced by the BLB. 
This yields $t$ 
estimates of both the parameter means and covariances (both with associated privacy variances) 
$\{ \tilde{\theta}_m, \vec{\sigma}^2_{\tilde{\theta}, m} \}_{m \in [t]}$ and
$\{ \tilde{\Sigma}_m, \vec{\sigma}^2_{\tilde{\Sigma}, m} \}_{m \in [t]}$.

\subsection{Algorithm Step 3: Get Final Parameter Estimates and Confidence Intervals Via Postprocessing}
\label{subsection:algorithm_step_3}

We now use our $t$ DP mean and covariance estimates $\{ \tilde{\theta}_m, \tilde{\Sigma}_m \}_{m \in [t]}$
to get a single DP estimate of the mean and covariance, which we'll call $\tilde{\theta}$ and $\tilde{\Sigma}$.

\paragraph{Final Parameter Estimates}
\label{subsubsection:final_parameter_estimates}

We start by presenting a multivariate version of the precision weighting argument
which gives the minimal covariance way to combine a 
set of unbiased estimators. 

\begin{theorem}
	\label{theorem:precision_weighting}
	For a parameter $\tau$, say we are given a series of independent estimates $\{ \hat{\tau}_m \}_{m \in [t]}$
	such that $\E \left( \hat{\tau}_m \right) = \tau$ and $\text{Cov} \left( \hat{\tau}_m \right) = S_m$
	for some positive definite $S_m$.
	Then the minimum covariance unbiased linear weighting 
	of the $\{ \hat{\tau}_m \}_{m \in [t]}$ is given by 
	\[
    \hat{\tau} = \left( \Sigma_{m=1}^{t} S_m^{-1} \right)^{-1} \left( \Sigma_{m = 1}^{t} S_m^{-1} \hat{\tau}_m \right)
    ,\]
	which has 
	$\E \left( \hat{\tau} \right) = \tau$
	and 
	$\text{Cov} \left( \hat{\tau} \right) = \left( \Sigma_{m=1}^{t} S_m^{-1} \right)^{-1}$.
\end{theorem}

Specifically, by ``minimum covariance'' we mean that any other unbiased linear weighting $\hat{\tau}'$ of the $\{ \hat{\tau}_m \}_{m \in [t]}$ will have 
$\text{Cov} \left( \hat{\tau} \right) \preceq \text{Cov} \left( \hat{\tau}' \right)$. 

Recall from Theorem~\ref{theorem:mean_est_form} that, with high probability, both our $\{\tilde{\theta}_m\}_{m \in [t]}$ and $\{\tilde{\Sigma}_m\}_{m \in [t]}$ 
are sets of estimates which are independent, unbiased, and have a known covariance structure. 
Thus, they both meet the criteria to be combined into precision-weighted estimators in the style of 
Theorem~\ref{theorem:precision_weighting}.
Because the precision weighting step is a function of the DP estimates,
it retains the privacy guarantees from Step 2 by the postprocessing property of zCDP (Lemma~\ref{lemma:zcdp_postprocessing}). 

We start by using the precision-weighting idea to find an estimator for the covariances, $\tilde{\Sigma}$.
Unlike the mean estimation setting, where our goal is to produce an unbiased private estimator,
our goal for our private covariance estimator is to find a private estimator that will
reliably overestimate the empirical covariance, and thus yield valid confidence intervals. 
To get a sense for why this is necessary, consider the one-dimensional case where we have a sample variance $s$
and a privatized version $\tilde{s}$.
In the non-private setting, we simply use the estimate $s$ to calculate
confidence intervals, but in order to use $\tilde{s}$ for confidence intervals we need to understand the 
relationship between $s$ and $\tilde{s}$. An important first step is to ensure that there is no 
clipping in the construction of $\tilde{s}$, so we know that $\tilde{s}$ equals $s$ plus zero-mean noise; this is the 
same motivation as in the mean estimation case. 
An extra complication for variances is that zero-mean noise addition has an asymmetric impact on confidence interval 
coverage; if we were to use $\tilde{s}$ for confidence intervals, we would underestimate the true variance 
$50\%$ of the time.
We can avoid this problem by increasing $\tilde{s}$ to the
point that we know, with high probability, that it is at least as large as $s$.  

\begin{theorem}
	\label{theorem:tilde_Sigma}
	Given covariance estimates and privacy variances $\{ \tilde{\Sigma}_m, \vec{\sigma}^2_{\Sigma, m} \}_{m \in [t]}$,
    let $\tilde{S}_m \in \R^{d'}$ be the flattened version of $\tilde{\Sigma}_m$.
    We can construct a precision-weighted estimator $\tilde{S}$:
	\[
		\tilde{S} := \frac{ \sum_{m=1}^{t} \tilde{S}_m / \vec{\sigma}_{\Sigma,m}^2 }{ \sum_{m=1}^{t} 1 / \vec{\sigma}_{\Sigma,m}^2 }.
    \]

    Let $\tilde{\Sigma}'$ be the unflattened $d \times d$ version of $\tilde{S}$ and $b$ be the unflattened $d \times d$
    matrix where $b^2_{ij} = \text{Var} \left( \tilde{\Sigma}'_{ij} \right)$ (i.e.\ the diagonal values of the covariance matrix of the flattened 
    precision-weighted estimator). For $\beta^{ub} \in (0,1)$, define
    \[
        c = \min_{\epsilon \in (0, 1/2]} (1+\epsilon) \left( 2 \max_{i \in [d]} \norm{b_{\cdot, j}}_2 + 
                                    \frac{6 \sqrt{\log d}}{\log(1+\epsilon)} \max_{i,j \in [d] \times [d]} |b_{ij}| \right) + 
                                    \sqrt{ \frac{\ln(1 / \beta^{ub})}{ 4 \max_{ij} b_{ij}^2 } }.
    \]

	Then, for 
    $\tilde{\Sigma} = \tilde{\Sigma}' + c I_{d}$
	we have
	$
        \Pr \left( \hat{\Sigma} \preceq \tilde{\Sigma} \right) \geq 1 - \beta^{\tilde{\Sigma}} - \beta^{ub}. 
    $
\end{theorem}

We use a similar strategy to convert our private parameter estimates $\{ \tilde{\theta}_m \}_{m \in [t]}$ into a
final private parameter estimate $\tilde{\theta}$. Because we want an unbiased estimator of $\hat{\theta}$,
we don't need to be conservative like we did with the covariance; we simply use the precision-weighted estimator.

\begin{theorem}
    \label{thm:tilde_theta}
    Given parameter estimates and privacy variances $\{ \tilde{\theta}_m, \vec{\sigma}^2_{\theta, m} \}_{m \in [t]}$,
    we define the precision-weighted estimator $\tilde{\theta}$ as
    \[
		\tilde{\theta} := \frac{ \sum_{m=1}^{t} \tilde{\theta}_m / \vec{\sigma}_{\theta,m}^2 }{ \sum_{m=1}^{t} 1 / \vec{\sigma}_{\theta,m}^2 }.
    \]
    This estimator has expectation $\hat{\theta}$ and covariance $\Sigma_{\tilde{\theta}} = \frac{1}{ \sum_{m=1}^{t} 1/\vec{\sigma}_{\theta,m}^2 }I_d$.
    In particular, we say that 
    \[
        \Pr \left( \tilde{\theta} \sim N \left( \hat{\theta}, \Sigma_{\tilde{\theta}} \right) \right) \geq 1 - \beta^{\tilde{\Sigma}} - \beta^{ub} - \beta^{\tilde{\theta}}.
    \]
\end{theorem}

\paragraph{Confidence Region}
\label{subsubsection:confidence_intervals}
We now have private estimates $\tilde{\theta}$ and $\tilde{\Sigma}$
such that $\tilde{\theta} \sim N\left( \hat{\theta}, \Sigma_{\tilde{\theta}} \right)$ and 
$\hat{\Sigma} \preceq \tilde{\Sigma}$ with probability $1 - \beta^{\tilde{\Sigma}} - \beta^{ub} - \beta^{\tilde{\theta}}$.
Going back to Assumption~\ref{assump:user_heavier_tailed}, we also have a distribution $Q_{\hat{\theta}}$ 
we assume to be heavier tailed than that of $\hat{\theta}$. 
We can represent our approximation of the sampling distribution of our estimator as the compound distribution
$
	Q_{\hat{\theta}} \left( \hat{\theta} + N \left(0, \Sigma_{\tilde{\theta}} \right), \tilde{\Sigma} \right).
$

\begin{theorem}[Confidence Region (valid with high probability)]
    \label{theorem:general_cr_high_prob}
    Let $Z$ be a $d$-dimensional random variable such that 
    $Z \sim Q_{\hat{\theta}} \left( \hat{\theta} + N \left(0, \Sigma_{\tilde{\theta}} \right), \tilde{\Sigma} \right)$.
    Suppose $C$ is a $d$-dimensional ellipsoid such that 
    $\Pr( Z \in C ) \geq 1 - \alpha$ for some $\alpha \in (0,1)$.
    Then, with probability $1 - \beta^{\tilde{\Sigma}} - \beta^{ub} - \beta^{\tilde{\theta}}$:
    \[ \Pr \left( \theta \in C \right) \geq 1 - \alpha. \]
\end{theorem}

It is always trivial to find such an ellipsoid $C$ (e.g.\ take $C = \R^d$), but finding one 
with coverage close to $1-\alpha$ analytically could be difficult in general. 
However, in typical scenarios it is likely to be much more nicely behaved. 
Most notable is the case where $Q_{\hat{\theta}}$ is multivariate Gaussian (e.g.\ the analyst is comfortable assuming that the 
CLT has kicked in for the BLB estimates).
The resulting compound distribution is 
$N \left(\tilde{\theta}, \tilde{\Sigma} + \Sigma_{\tilde{\theta}} \right)$; a Gaussian random variable with a 
Gaussian random variable as its location parameter is still Gaussian. 
This becomes even simpler if the analyst is interested in univariate confidence intervals, in which case they can use the fact that
$\forall j \in [d]: \tilde{\theta}_j \sim N \left( \hat{\theta}_j, \tilde{\Sigma}_{j,j} + \Sigma_{\tilde{\theta}_{j,j}} \right)$,
and can calculate confidence intervals directly from the CDF of the univariate Gaussian.
In more complicated scenarios, an analyst can always get approximate quantiles using Monte Carlo methods.

\subsection{Full Algorithm Statement}
In Algorithm~\ref{alg:gvdp}, we finally present our algorithm in whole. 
We omit some hyperparameters in the subroutines to make it easier to focus on the core pieces that change between them.
Let $\xi: a_{1:k} \mapsto \frac{1}{k} \sum_{i=1}^{k} a_i$ and either
$\xi': a_{1:k} \mapsto \text{Cov} \left( \{a_i\}_{i \in [k]} \right)$
or 
$\xi': a_{1:k} \mapsto \diag \left( \text{Cov} \left( \{a_i\}_{i \in [k]} \right) \right)$,
depending on whether the analyst desires a joint confidence region or separate confidence intervals.

\begin{algorithm}[h!]
    \scriptsize
    \caption{General Valid DP (GVDP)}
    \label{alg:gvdp}
    \hspace*{\algorithmicindent} \textbf{Input:} data set $X \in \R^{n \times m}$, 
                                                 estimator $\hat{\theta}: \mathcal{X}^n \rightarrow \R^d$ 
                                                 families of distributions $Q_{\hat{\theta}}, Q_{\hat{\Sigma}}$,
                                                 privacy budgets $\rho^{\tilde{\theta}}, \rho^{\tilde{\Sigma}} > 0$,
                                                 failure probabilities $\beta^{\tilde{\theta}}, \beta^{\tilde{\Sigma}}, \beta^{ub} \in (0,1)$ \\ 
    \hspace*{\algorithmicindent} \textbf{Output:} parameter estimate $\tilde{\theta}$
                                                  and associated confidence intervals/region $C$
                                                  which satisfy $\left( \rho^{\tilde{\theta}} + \rho^{\tilde{\Sigma}} \right)$-zCDP
                                                  and have desired unbiased/validity properties with probability $1 - \beta^{\tilde{\theta}} - \beta^{\tilde{\Sigma}} - \beta^{ub}$
    \begin{algorithmic}[1] 
        \Procedure{\gvdp}{$X, \hat{\theta}, Q_{\hat{\theta}}, Q_{\hat{\Sigma}}, \rho^{\tilde{\theta}}, \rho^{\tilde{\Sigma}}$}
            \State $\{ \hat{\Sigma}^{BLB}_i \}_{i \in [k]} = \text{\blb} \left( X, \hat{\theta}, \xi', \hdots \right)$ \Comment{Algorithm~\ref*{appendix:alg:blb} -- get BLB estimates of parameter covariance} \label{ln:gvdp:blb_cov}
            \State $\{ \hat{\theta}^{BLB}_i \}_{i \in [k]} = \text{\blb} \left( X, \hat{\theta}, \xi, \hdots \right)$ \Comment{Algorithm~\ref*{appendix:alg:blb} -- get BLB estimates of parameter means} \label{ln:gvdp:blb_mean}
            \State $\{\tilde{\Sigma}_m\}_{m \in [t]} = \text{\mvmeaniter} \left( \{ \hat{\Sigma}^{BLB}_i \}_{i \in [k]}, \hdots, Q_{\hat{\Sigma}}, \hdots, \rho^{\tilde{\Sigma}}, \beta^{\tilde{\Sigma}} \right) $ \Comment{Algorithm~\ref*{appendix:alg:hd_mean_est} -- privately estimate parameter covariance at $\rho^{\tilde{\Sigma}}$-zCDP level} \label{ln:gvdp:cp_cov}
            \State Combine $\{\tilde{\Sigma}_m\}_{m \in [t]}$ via precision-weighting to get $\tilde{\Sigma}$ \Comment{Theorem~\ref{theorem:precision_weighting}} \label{ln:gvdp:prec_weight_cov}
            \State $\{\tilde{\theta}_m\}_{m \in [t]} = \text{\mvmeaniter} \left( \{ \hat{\theta}^{BLB}_i \}_{i \in [k]}, \hdots, Q_{\hat{\theta}}, \hdots, \rho^{\tilde{\theta}}, \beta^{\tilde{\theta}} \right) $ \Comment{Algorithm~\ref*{appendix:alg:hd_mean_est} -- privately estimate parameter means at $\rho^{\tilde{\theta}}$-zCDP level} \label{ln:gvdp:cp_mean}
            \State Combine $\{\tilde{\theta}_m\}_{m \in [t]}$ via precision-weighting to get $\tilde{\theta}$ \Comment{Theorem~\ref{theorem:precision_weighting}} \label{ln:gvdp:prec_weight_mean}
            \State Use $\tilde{\theta}, \tilde{\Sigma}, Q_{\hat{\theta}}, \text{ and } \beta^{ub}$ to get 
                   confidence intervals/region $C$. \Comment{Theorem~\ref{theorem:general_cr_high_prob}} \label{ln:gvdp:get_cis}
            \State \Return $\{ \tilde{\theta}, C \}$
        \EndProcedure
    \end{algorithmic}
\end{algorithm}

In summary, our private estimates $\tilde{\theta}$ and $\tilde{\Sigma}$ have statistical guarantees relative to the true parameters 
of the sampling distribution of $\hat{\theta} \sim G(\theta, \Sigma)$ via the following lines of reasoning:
\[
    \E \left( \tilde{\theta} \right)
    \underset{ \text{CoinPress + postprocessing} }{ \stackrel{w.h.p.}{=} }
    \E \left( \hat{\theta}^{BLB} \right)
    \underset{ \text{Assumption~\ref{assump:blb_approx}} }{=}
    \E \left( \hat{\theta} \right)
    = \theta
\]
\[
    \tilde{\Sigma}
    \underset{ \text{CoinPress + postprocessing} }{ \stackrel{w.h.p.}{\succeq} }
    \hat{\Sigma}^{BLB}
    \underset{ \text{Assumption~\ref{assump:blb_approx}} }{ \succeq }
    \Sigma
\]

%% file: empirical_evaluation.tex
\section{Empirical Evaluation}
\label{section:empirical_evaluation}

We provide empirical demonstrations of our core result, showing that we can produce unbiased parameter estimates 
and valid confidence intervals when the requisite assumptions hold. For every evaluation, we aim to get 
valid confidence intervals for each element of the parameter vector rather than a single valid confidence region, as we expect this to be the dominant use case in practice. 

All results 
satisfy zCDP at the $\rho = 0.1$ level and, inside the GVDP algorithm, we always run CoinPress for $t=5$ iterations.
Additionally, we assume that the analyst chooses bounds that satisfy 
Assumptions~\ref{assump:initial_parameter_bounds} and \ref{assump:initial_covariance_bounds},
but are larger than the tightest possible bounds by a factor of $\approx 100$.
For example, if the analyst had an estimator with a $N(\mu = 1, \sigma^2 = 1)$ sampling distribution
we assume their prior knowledge to be that $\mu \in [-100, 100]$ and $\sigma^2 \leq 100$.  
Additional empirical results, including comparisons to existing methods, are described in Section~\ref*{appendix:empirical_results}.

\paragraph{OLS Regression Demonstration}
\label{section:ols_regression_demo}

We begin by testing parameter estimation for a properly specified OLS model with $d=5$ parameters of interest.
In a single iteration of our experiment, we generate data from a linear model $y = X \beta + \epsilon$ 
with Gaussian covariates, Gaussian error, 
and correlation structure such that the effective rank of the resulting data is $\approx d - 1$.
We increase the underlying noise in the data as $n$ increases such that the non-private confidence intervals 
are essentially constant across values of $n$. This allows us to better demonstrate the effects of changes in $k$ and $n$ on 
our algorithm's performance. 
We then privately estimate the values of the $d$ coefficients and their associated 
standard errors. We run this entire experiment 100 times. We assume the sampling distribution of the coefficients is multivariate Gaussian
and imagine that the user sets all upper bounds $\approx 100$ times larger than the 
tightest possible upper bounds. 
We present these results in Figure~\ref{fig:ols_regression}.

\begin{figure}[h!]
    \captionsetup{font=footnotesize,labelfont=footnotesize}
    \begin{subfigure}[t]{.49\textwidth}
        \centering
        \includegraphics[width=\textwidth,height=4cm]{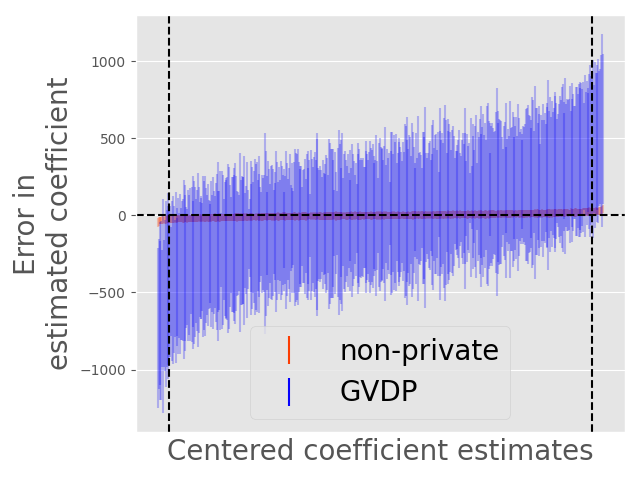}
        \caption{$n = 10{,}000, k = 250$}
    \end{subfigure}
    \begin{subfigure}[t]{.49\textwidth}
        \centering
        \includegraphics[width=\textwidth,height=4cm]{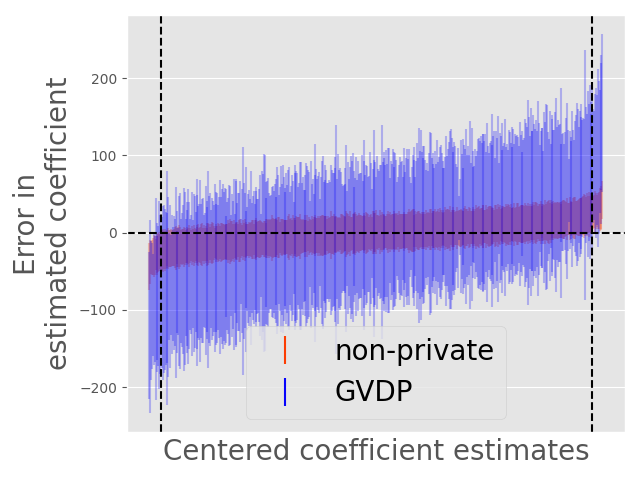}
        \caption{$n = 10{,}000, k = 1{,}000$}
    \end{subfigure}
    \begin{subfigure}[t]{.49\textwidth}
        \centering
        \includegraphics[width=\textwidth,height=4cm]{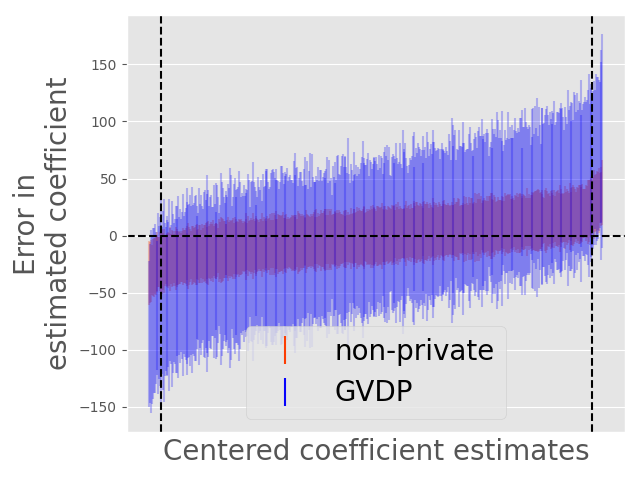}
        \caption{$n = 50{,}000, k = 1{,}000$}
    \end{subfigure}
    \begin{subfigure}[t]{.49\textwidth}
        \centering
        \includegraphics[width=\textwidth,height=4cm]{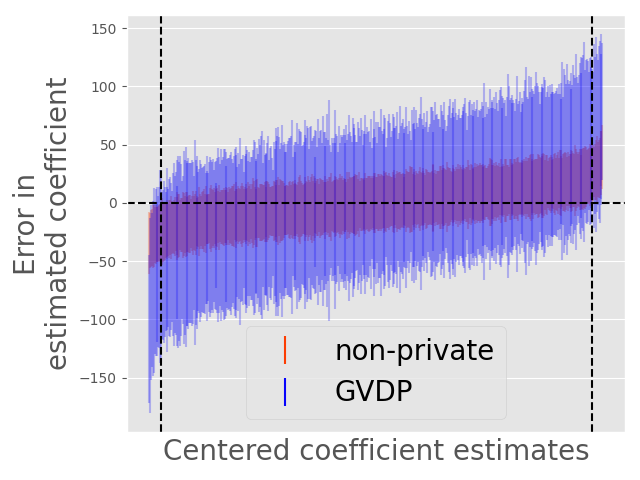}
        \caption{$n = 100{,}000, k = 1{,}000$}
    \end{subfigure}
    \caption{OLS: Distribution of coefficient estimates and 95\% confidence intervals.}
    \label{fig:ols_regression}
\end{figure}

Each plot consists of coefficient estimates centered around their true values and presented in increasing order, 
with vertical bars representing the 95\% confidence interval for that estimate. 
We expect properly calibrated confidence intervals to cross the x-axis at the vertical dotted black lines, placed 
at the $2.5^{th}$ and $97.5^{th}$ quantiles, which is the behavior we observe 
in each plot. 

Additionally, Figure~\ref{fig:ols_regression} demonstrates the principle that the noise due to privacy in our algorithm scales with 
$k$ rather than $n$. 
The private confidence intervals are significantly tighter in plot (b) than in plot (a), 
while plots (c) and (d) are essentially identical. 

One complicating factor to this story is that plot (c) looks better than plot (b), even though they use the same $k$.
This is because the variance of the BLB estimates will tend to decrease as $\frac{n}{k}$ increases, up to the point 
where $\frac{n}{k}$ is large enough that the BLB estimates have converged to the sampling distribution of the estimator.
Figure~\ref{fig:blb_estimates} shows the distribution of the BLB estimates of two of our estimated coefficients at the different 
levels of $\frac{n}{k}$ used in plots (b) and (c). 

\begin{figure}[h!]
    \captionsetup{font=footnotesize,labelfont=footnotesize}
    \begin{subfigure}[t]{0.49\textwidth}
        \centering
        \includegraphics[width=\textwidth, height=4cm]{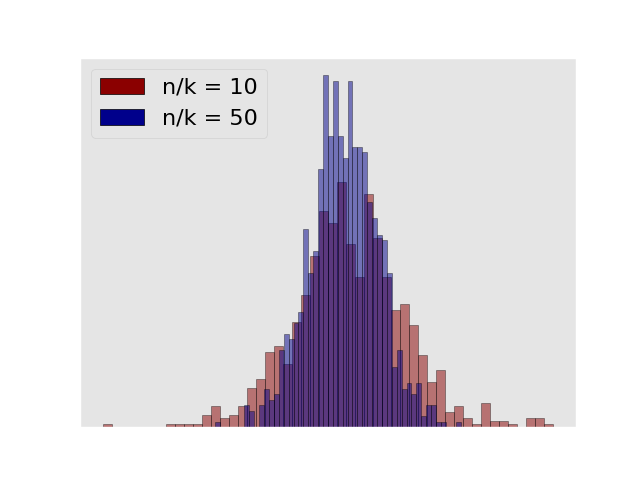}
        \caption{Coefficient 1}
    \end{subfigure}
    \begin{subfigure}[t]{0.49\textwidth}
        \centering
        \includegraphics[width=\textwidth, height=4cm]{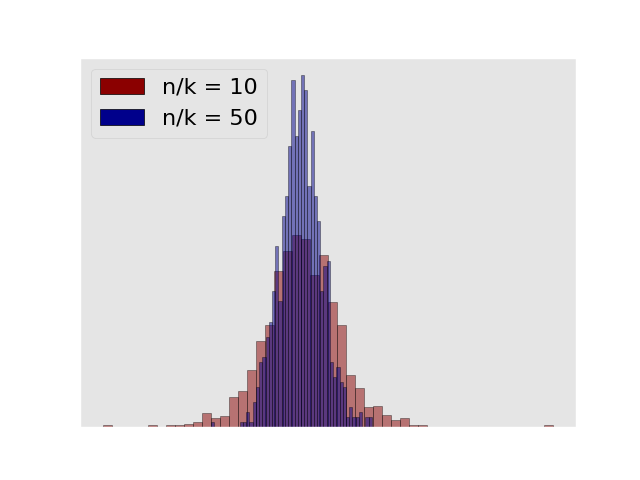}
        \caption{Coefficient 3}
    \end{subfigure}
    \caption{OLS: BLB estimates from a single run}
    \label{fig:blb_estimates}
\end{figure}

Ideally, the analyst should attempt to choose $k$ to be as large as possible, subject to the constraint that the 
BLB estimates, calculated on subsets of size $\frac{n}{k}$, converge to the true sampling distribution of the estimator.

One possible strategy for doing this is for the analyst to test the BLB estimation for their estimator of interest 
on non-sensitive data and use a $k$ that performs well on the non-sensitive data. In particular, an analyst could 
start with $k = 1$ and increase $k$ until the distribution of BLB estimates starts to look substantially different 
(i.e.\ different beyond Monte Carlo error). 
Of course, this approach works only if the analyst can find or generate non-sensitive data that
are similar enough to their sensitive data that the BLB will perform similarly on each.
We discuss the issue further in Section~\ref*{appendix:subsection:choosing_k}.

\subsection{Comparison with UnbiasedPrivacy~\citep{EKST19}}
\label{subsection_comparison_with_ekst}

In Table~\ref{table:ekst_comparison}, we compare our algorithm to the UnbiasedPrivacy (UP) algorithm from~\citet{EKST19}.
UP is designed for estimators with a univariate Gaussian sampling distribution,
so we focus on that setting here.

While the goal of GVDP is to let users set very conservative bounds 
(which the algorithm improves upon)
and not clip any points in the aggregation step,
UP requires that users set reasonably good bounds upfront. 
UP never tightens bounds that are too loose, but it will attempt to 
bias-correct the results if the analyst's chosen bounds clip some of the data (which our 
algorithm does not). 

We test UP against GVDP across five scenarios using the implementation provided in the original paper. 
As suggested in \citet{EKST19}, we split the privacy budget evenly between the mean estimation 
task and the estimation of the proportion of points that are clipped.
We consider two cases in which we expect UP to perform bias correction, 
clipping the top 20\% and 10\% of the BLB estimates, and three in which we don't; 
clipping bounds set as tightly as possible with no clipping, bounds set three times larger than 
the true values of the parameter, and bounds set $1{,}000$ times larger than the true value. 

We run $1{,}000$ simulations, each of which involves
generating $n = 50{,}000$ data points $Y_i \sim N \left( 0, 250 \right)$ 
and using UP and GVDP to generate a private OLS estimator,
using $k = 500$ as the number of subsets for the BLB for each method.
In Table~\ref{table:ekst_comparison}, we provide the average $\ell_1$ error in the coefficient estimate, 
average standard error, and empirical 95\% confidence interval coverage for each method with various levels of clipping bound. 

\begin{table}
    \small 
    \centering 
    \setlength\tabcolsep{3pt} 
    \renewcommand{\arraystretch}{0.5} 
    \begin{tabular}{|c||c|c|c|c|}
        \hline
        \textbf{Bounds} & \textbf{Method} & \textbf{Avg Coef Err} & \textbf{Avg SE} & \textbf{CI Cov} \\ 
        \hline
        Top 20\% Clipped & UP & 0.222 & 0.0872 & 0.237 \\
        & GVDP & 0.007 & 0.207 & 0.970 \\
        \hline
        Top 10\% Clipped & UP & 0.080 & 0.108 & 0.932 \\
        & GVDP & 0.003 & 0.207 & 0.970 \\
        \hline
        Tightest bounds (no clipping) & UP & 0.129 & 0.145 & 0.973 \\
        & GVDP & 0.001 & 0.208 & 0.975 \\
        \hline
        3 times too large & UP & 0.004 & 0.370 & 0.968 \\
        & GVDP & 0.004 & 0.218 & 0.973 \\
        \hline
        $1{,}000$ times too large & UP & 3.239 & 121.090 & 0.950 \\
        & GVDP & 0.005 & 0.701 & 0.961 \\ 
        \hline
    \end{tabular}
    \caption{Comparison of UP and GVDP: Average Coefficient Estimate, Average Standard Error, and Empirical 
    Coverage of 95\% Confidence Intervals}
    \label{table:ekst_comparison}
\end{table}

In the 10\% clipping and tight bound (no clipping) settings, UP mostly delivers as advertised; it gives
(approximately) valid confidence intervals and does so with smaller standard errors than
GVDP does under any bounds.
However, it is unable to achieve this when the top 20\% of BLB estimates are clipped, yielding highly biased coefficient
estimates and poor CI coverage. Moreover, even in the 10\% and tight bound settings, UP
does not appear to achieve truly unbiased coefficient estimates. We believe this is because of 
the error in the bias correction step of UP, created when privately estimating the proportion of clipped data points.

In the cases where the bounds are too conservative, we see that GVDP outperforms UP, 
as GVDP is designed to improve conservative bounds whereas UP is not.
It's notable that in the 10\% and 20\% clipping where GVDP gives no guarantees, 
it appears to provide unbiased coefficient estimates and valid CI coverage. 
This is because of GVDP's variance estimates are conservative by design (to ensure 
valid coverage with high probability), so even when the initial data bounds are set too narrowly 
it is possible that GVDP's overly conservative variances compensate inside of CoinPress such that 
no BLB estimates end up being clipped.

%% file: discussion.tex
\section{Discussion}
\label{section:discussion} 

We believe that whether or not our method (GVDP) is effective relative to 
other approaches will generally 
come down to a few different factors. 

First, we suggest that GVDP be 
considered primarily when the analyst is not confident in their ability to set
``good'' bounds on their underlying data domain $\mathcal{X}$ for their given estimator.
This is a function of both analyst knowledge of $\mathcal{X}$ and the properties of 
their estimator, as some estimators will be robust even if the analyst sets bounds that clip 
small proportions of the data, while others will not be
(see our demonstration in Section~\ref{subsection:problem_demonstration}).

Second, 
GVDP is likely to work well 
only for reasonably large $n$. Recall that GVDP partitions the 
data $X$ into $k$ subsets of size $\frac{n}{k}$ and bootstrapping the estimator over each subset. 
This creates a tradeoff between the plausibility of our assumptions and the required noise addition 
to satisfy DP.
As $k$ increases, the sensitivity of our aggregator decreases and so too does the variance of the noise 
in our privacy mechanism.

However, we require that the BLB estimates be good estimates of the 
sampling distribution of our estimator in the non-private setting (Assumption~\ref{assump:blb_approx}).
Similar to other bootstrap methods, the BLB's 
guarantees are asymptotic 
(see Section 3 of~\citet{KTSJ14}), and so it is difficult to know how reasonable 
Assumption~\ref{assump:blb_approx} is when $\frac{n}{k}$ is small.

Finally, we believe GVDP can be useful in settings where the dimension of the estimand of interest is significantly 
lower than the dimension of the data. 
Higher dimensional domains create two potential problems for differentially private estimation that are not 
present in non-private estimation. First, setting clipping bounds, generally speaking, gets 
more difficult as the dimension increases. Second, the function sensitivity, and thus variance of the noise 
in our privacy mechanism, increases with the dimension of the function's input domain 
(again see Lemma~\ref*{appendix:lemma:gaussian_mechanism}).
The function being privatized in GVDP takes the BLB parameter estimates
as input rather than the full data, so these two problems are minimized if the estimand is 
low-dimensional relative to the original data.

%% file: supplement.tex
\section{Definitions}
\label{appendix:definitions}

\subsection{Differential Privacy}
\label{appendix:subsection:differential_privacy_definitions}
We begin with an introduction to the core definitions of DP.

\begin{definition}[Neighboring data sets]
    \label{appendix:defn:neighboring_data_sets}
    Let $\mathcal{X}$ be a data universe and $D, D' \in \mathcal{X}^n$.
    We say that $D, D'$ are neighboring if 
    \[ \max \left( |D \setminus D'|, |D' \setminus D| \right) = 1. \]
    We also define the set of all neighboring data sets as 
    \[ \mathcal{D}_n = \{ (D, D') \in \mathcal{X}^n \times \mathcal{X}^n : D, D' \text{ are neighbors} \}. \]
\end{definition}

\begin{definition}[R{\'e}nyi divergence~\citep{Ren61}]
    \label{appendix:defn:renyi_divergence}
    Let $P,Q$ be probability measures over a measurable space $(\Omega, \Sigma)$. Then we define 
    the $\alpha$-R{\'e}nyi divergence between $P,Q$ as  
    \[ H_{\alpha} \left( P \Vert Q \right) = \frac{1}{\alpha-1} \ln \int_{\Omega} P(x)^{\alpha} Q(x)^{1-\alpha}dx. \]
\end{definition}

\begin{definition}[Global Function Sensitivity]
    \label{appendix:defn:global_sensitivity}
    Let $\mathcal{X}$ be a data domain, 
    $\gamma: \mathcal{X}^n \rightarrow \R^d$, and
    $\mathcal{D}_n$ be the set of all neighboring data sets as in Definition~\ref{appendix:defn:neighboring_data_sets}.
    Then we write the global sensitivity of $\gamma$ with respect to a distance metric $d$ as 
    \[ GS_{d}(\mathcal{X}^n, \gamma) = \max_{D, D' \in \mathcal{D}_n} d \left(\gamma(D), \gamma(D') \right). \]
\end{definition}

Algorithms can be made to respect DP in a variety of ways, but the most common way (as well as 
the approach we use in this work) is via an 
\emph{additive noise mechanism}. This just entails running the algorithm as one would normally, and then 
adding random noise scaled relative to the algorithm's sensitivity.

Throughout this work, we use a popular additive noise mechanism called the \emph{Gaussian mechanism}.
\begin{lemma}[Gaussian Mechanism]
    \label{appendix:lemma:gaussian_mechanism}
    Let $f : \mathcal{X}^n \rightarrow \R^d$ have global $\ell_2$ sensitivity $GS_{\ell_2}(\mathcal{X}^n, f)$.
    Then the Gaussian mechanism 
    \[ \mathcal{M}_f(D) = f(D) + N \left( 0, \left( \frac{GS_{\ell_2}(\mathcal{X}^n, f)}{\sqrt{2 \rho}} \right)^2 I_{d} \right) \]
    satisfies $\rho$-zCDP.
\end{lemma} 

Note that it is often necessary to bound the data domain $\mathcal{X}$ to ensure that 
$GS_{\ell_2}(\mathcal{X}^n, f) < \infty$.
For example, let $\mathcal{X} = \R^d$, $D = (D_1, \hdots, D_n)$ with $D_i \in \mathcal{X}$, 
and $f: \R^{n \times d} \rightarrow \R^d$ be such that $f(D) = n^{-1} \sum_{i=1}^{n} D_i$. 
If we let $D' = (\infty, D_2, \hdots, D_n)$, then $D,D'$ are neighbors (they differ only in the first element),
but $\Vert f(D) - f(D') \Vert_2 = \infty$. 
If instead $\mathcal{X} = [0,1]^d$, then the $D,D'$ that induce the largest difference in $f$ are 
$D = (\vec{1}, D_2, \hdots, D_n)$ and $D' = (\vec{0}, D_2, \hdots, D_n)$. In this scenario,
$\Vert f(D) - f(D') \Vert_2 = \Vert n^{-1} (\vec{1} - \vec{0}) \Vert_2 = n^{-1} \sqrt{d}$, 
and thus $GS_{\ell_2}(\mathcal{X}^n, f) = n^{-1} \sqrt{d}$.

These bounds must be set without looking at the particular $D_i$, and 
are generally chosen by a data analyst based on public metadata and/or their beliefs about the 
data-generating process. 

\subsection{Statistical Inference}
\label{subsection:statistical_inference_definitions}

This need to bound $\mathcal{X}$ introduces complications for doing statistical inference under DP, while 
maintaining the types of guarantees we often want from non-private estimators. 
We focus specifically on unbiased estimators and valid confidence sets.

\begin{definition}[Unbiased Estimator]
    \label{appendix:defn:unbiased_estimator}
    Let $\theta \in \R^d$ be a model parameter we wish to estimate. We collect data $D \sim \mathcal{D}$ and estimate 
    $\theta$ with a random variable $\hat{\theta}: \mathcal{D} \rightarrow \R^m$. We say that $\hat{\theta}$ is an unbiased 
    estimator of $\theta$ if 
    \[ \E \left( \hat{\theta}(D) \right) = \theta, \]
    with randomness taken over the sampling of $D \sim \mathcal{D}$, 
    as well as any other randomness in $\hat{\theta}$.   
\end{definition}

Many applied statisticians, particularly those interested in estimating 
causal effects using linear models, prize unbiased parameter estimation 
and are willing to sacrifice on other fronts to achieve it. 
For example, the standard OLS estimator (which is the minimum-variance unbiased estimator 
under the assumptions of the Gauss-Markov theorem) is used for estimating parameters 
of a linear regression model in favor of other biased estimators, such as 
the James-Stein estimator~\citep{Ste56, SJ61}, which dominate it 
in terms of $\ell_2$ error of the parameter estimates.   

\begin{definition}[Confidence Set]
    \label{appendix:defn:confidence_set}
    Let $\theta \in \R^d$ be a model parameter we wish to estimate using data $D \sim \mathcal{D}$.
    For arbitrary $\alpha \in [0,1]$, a $(1-\alpha)$-level confidence set for $\theta$ is a 
    random set $S \subseteq \R^d$ such that
    \[ \Pr \left( \theta \in S \right) = 1 - \alpha, \]
    with randomness taken from the sampling of $D$ and any other randomness in the 
    construction of $S$.  
\end{definition}

Ideally, we would be able to find a perfectly-calibrated confidence set, where the 
coverage probability (i.e. $\Pr(\theta \in S)$) is exactly $1 - \alpha$.
However, this is often impossible to compute exactly and so practitioners 
tend to default to being overly conservative instead.
In this setting, we require $\Pr \left( \theta \in S \right) \geq 1 - \alpha$
and call such an $S$ a \emph{valid confidence set}.
In this work, we focus on \emph{confidence regions}, which are 
contiguous confidence sets, and occasionally \emph{confidence intervals}, which 
are univariate confidence regions.

We can simplify the general problem of constructing confidence sets 
by restricting our attention to estimators 
whose sampling distribution belongs to a symmetric multivariate location-scale family.

\begin{definition}[Location-Scale Family]
	A set of probability distributions is a location-scale family if any
	density $f(x; \mu, \Sigma)$ in the set is written as 
	$f(x; \mu, \Sigma) = c \vert \Sigma \vert^{-1/2} \exp \left( -\frac{1}{2} (x-\mu)^T \Sigma^{-1} (x-\mu) \right)$.
	for some normalization constant $c$.
\end{definition}

Our restriction to location-scale families ensures that estimating the mean and 
(co)variance of the estimator is sufficient to characterize its distribution.

\section{Step 1: Bag of Little Bootstraps}
\label{appendix:bag_of_little_bootstraps}

This algorithm statement is adapted and simplified for our purposes; 
readers interested in the original version should consult~\cite{KTSJ14}.
We say that $\mathcal{X}$ is our data universe, $\mathcal{D}$ is a distribution 
over $\mathcal{X}$, and our realized data $X \in \R^{n \times m}$ are drawn from 
$\mathcal{D}^n$. For an arbitrary estimator $\hat{\theta}: \mathcal{X}^n \rightarrow \R^d$,
we define $\hat{\theta}(\mathcal{D}) = \E_{X \sim \mathcal{D}^n} \left( \hat{\theta}(X) \right)$.

\begin{algorithm}[H]
	\footnotesize
	\caption{Bag of little bootstraps (BLB)}
	\label{appendix:alg:blb}
	\hspace*{\algorithmicindent}\textbf{Input:} data set $X \in \R^{n \times m}$, 
												estimator $\hat{\theta}: \mathcal{X}^n \rightarrow \R^d$,
												estimator quality assessment $\xi$,
												$k$ number of subsets of partition,
												$r$ number of bootstrap simulations \\
	\hspace*{\algorithmicindent}\textbf{Output:} $k$ estimates of $\hat{\theta}(\mathcal{D})$ 
	\begin{algorithmic}[1]
		\Procedure{\blb}{$X, \hat{\theta}, k, r$}
			\State Randomly partition $X$ into $k$ subsets $\{X_i\}_{i \in [k]}$
			\For{$i \in [k]$}
				\State $b = \vert X_i \vert$
				\State $\{\hat{\theta}_{i,c}\}_{c \in [r]} = \varnothing$ 
				\For{$c \in [r]$}
					\State sample $(n_1, \hdots, n_b) \sim \text{Multinomial}(n, \mathbf{1}_b / b)$
					\State create $X_i^U \in \R^{n \times m}$ by including the $j^{th}$ element of 
						$X_i$ $n_j$ times
					\State $\hat{\theta}_{i,c} = \hat{\theta}(X_i^U)$
				\EndFor
				\State $\hat{\theta}_i = \xi \left( \{ \hat{\theta}_{i,c} \}_{c \in [r]} \right)$
			\EndFor 
			\State \Return $\{\hat{\theta}_i\}_{i \in [k]}$
		\EndProcedure
	\end{algorithmic}
\end{algorithm}

\section{Step 2: Differentially Private Estimation}
\label{appendix:section:general_mean_estimation}

\begin{definition}
	\label{appendix:defn:heavier_tails}
	Let $\mathcal{B}(\mu, \Sigma)$ and $\mathcal{C}(\mu,\Sigma)$ be families of distributions and 
	$B,C$ be random variables drawn from each such that $\E(B) = \E(C) = \mu$ and $\text{Cov}(B) = \text{Cov}(C) = \Sigma$. 
	Let $PSD_d$ be the set of all $d \times d$ PSD matrices.
	We say that $\mathcal{B}$ is \emph{heavier-tailed} than 
	$\mathcal{C}$
	if for all $\mu \in \R^d, \Sigma \in PSD_d, \text{ and }
	v \in \R^d \text{ such that } \Vert v \Vert_2 = 1$, then  
	\[ \Pr \left( v^T (B - \mu) \leq z \right) 
	   \leq \Pr \left( v^T (C - \mu) \leq z \right),
	 \]
	 for all $z > 0$.
\end{definition}

\subsection{Modified CoinPress Algorithm}
\label{appendix:subsection:modified_coinpress_algorithm}

\begin{algorithm}[H]
	\footnotesize
    \caption{Modified CoinPress}
    \label{appendix:alg:hd_mean_est}
    \hspace*{\algorithmicindent} \textbf{Input:} $X = (x_1, \hdots, x_k)$ from a distribution $D$ with mean $\mu$ and covariance $\Sigma$,
													$\tilde{\Sigma}$ such that $\Sigma \preceq \tilde{\Sigma}$, 
													$B_2(\tilde{\mu}_0,r_0)$ containing $\mu$, 
													family of distributions $Q_{X}(\cdot, \Sigma_{\mu})$ with heavier 
													tails than $D$, 
													number of iterations $t \in \mathbb{N}^+$, 
													zCDP privacy loss parameter $\rho > 0$, 
													failure probability $\beta > 0$ \\ 
    \hspace*{\algorithmicindent} \textbf{Output:} $t$ estimates of $\mu$ that jointly respect $\rho$-zCDP 
    \begin{algorithmic}[1]
        \Procedure{\mvmeaniter}{$X, \tilde{\mu}_0, r_0, \tilde{\Sigma}, Q, t, \rho, \beta$} 
		\State $S = \tilde{\Sigma}^{1/2}$
		\State $\tilde{\mu}_0 = S^{-1} \tilde{\mu}_0$
		\State $r_0 = \max \left( \diag \left( S^{-1} \right) \right) \cdot r_0$
		\State Define $\bar{X} \in \R^{k \times d}$ such that $\forall j \in [d], \forall m \in [k] : \bar{X}_{m,j} = \frac{1}{k} \sum_{m'=1}^{k} x_{m',j}$.
			   Note that each row $\bar{X}_{m,:}$ is equal to the $d$-dimensional empirical mean of $X$ 
		\State $X' = \left( X - \bar{X} \right) S^{-1}$ 
        \For {$m \in [t-1]$} \label{ln:hd_mean_est-loop}
			\State $(\tilde{\mu}_m, r_m, \sigma_m) = \mvmeanstep(X', \tilde{\mu}_{m-1}, r_{m-1}, Q_{X}(0, I_d), \tfrac{\rho}{2(t-1)}, \tfrac{\beta}{t})$ \Comment{Algorithm~\ref{appendix:alg:hd_mean_step}} \label{appendix:ln:hd_mean_est-mvm1}
		\EndFor
		\State $(\tilde{\mu}_t,r_t,\sigma_t) = \mvmeanstep(X', \tilde{\mu}_{t-1}, r_{t-1}, Q_{X}(0, I_d), \tfrac{\rho}{2}, \tfrac{\beta}{t})$ \label{appendix:ln:hd_mean_est-mvm2}
		\State $\forall m \in [t]: \tilde{\mu}_m = \left( S \tilde{\mu}_m \right) + \bar{\mu}_{1,:}$ \Comment{convert mean estimates to proper scale} \label{appendix:ln:hd_mean_est-scale_mu}
		\State $\forall m \in [t]: \vec{\sigma}^2_{m} = \diag \left( S \sigma_m \right)^2$ \Comment{convert private noise variances to proper scale} \label{appendix:ln:hd_mean_est-scale_sigma}
		\State \Return $\{ (\tilde{\mu}_m, \vec{\sigma}^2_m)\}_{m \in [t]}$ \label{appendix:ln:hd_mean_est-return}
        \EndProcedure
    \end{algorithmic}
\end{algorithm}

\subsection{Modified CoinPress Algorithm - One Step Improvement}
\label{appendix:subsection:modified_coinpress_algorithm_one_step}

\begin{algorithm}[H]
	\footnotesize
    \caption{One Step Private Improvement of Mean Ball}
    \label{appendix:alg:hd_mean_step}
    \hspace*{\algorithmicindent} \textbf{Input:} $x = (x_1, \hdots, x_k)$ from a distribution with mean $0$ and covariance 
												 with smaller L{\"o}wner order than $I_d$, 
												 $B_2(\tilde{\mu}, r)$ containing $0$, 
												 family of distributions $Q_{X}(\cdot, I_d)$,
												zCDP privacy loss parameter $\rho_m > 0$, failure probability $\beta_m > 0$ \\
    \hspace*{\algorithmicindent} \textbf{Output:} A $\rho_s$-zCDP ball $B_2(\tilde{\mu}',r')$ and scale of the privatizing noise $\sigma$
    \begin{algorithmic}[1]
        \Procedure{\mvmeanstep}{$\hat{M}, \tilde{\mu}, r, Q_X, \rho_m, \beta_m$} 
			\State $\beta_s = \beta_m / 2$
			\State Let $R \sim Q_{X}(0, I_d)$ \label{appendix:ln:hd_mean_step-R}
            \State Set $\gamma_1$ such that $\Pr \left( \Vert R \Vert_2 > \gamma_1 \right) \leq \frac{\beta_s}{k}$ \label{appendix:ln:hd_mean_step-gamma1}
            \State Set $\gamma_2$ such that $\Pr \left( \Vert R \Vert_2 > \gamma_2 \right) \leq \beta_s$ \label{appendix:ln:hd_mean_step-gamma2} 
			\State Project each $x_i$ into $B_2(\tilde{\mu}, r + \gamma_1)$.\label{appendix:ln:hd_mean_step-trunc}
			\State $\Delta = 2(r+\gamma_1) / k$. \label{appendix:ln:hd_mean_step-sens}
			\State $\sigma = \frac{\Delta}{\sqrt{2\rho_s}}$ \label{appendix:ln:hd_mean_step_sigma}
			\State Compute $\tilde{\mu}' = \frac{1}{k}\sum_i x_i + Y$, where $Y \sim \text{N}\left(0,\sigma^2 I_{d}\right).$ \label{appendix:ln:hd_mean_step-gm}
			\State $r' = \gamma_2 \sqrt{ \frac{1}{k} + \frac{2 (r+\gamma_1)^2}{k^2 \rho_s} }$ \label{appendix:ln:hd_mean_step_rp}
			\State \Return $(\tilde{\mu}',r',\sigma)$. \label{appendix:ln:hd_mean_step-return}
        \EndProcedure
    \end{algorithmic}
\end{algorithm}

\subsection{Privacy Analysis of Algorithm~\ref{appendix:alg:hd_mean_est}}
\label{appendix:subsection:proof:privacy_analysis}

\begin{theorem}[Modified CoinPress Privacy Statement]
	\label{appendix:thm:modified_coinpress_privacy_statement}
	Algorithm~\ref{appendix:alg:hd_mean_est} produces 
	$t$ estimates of $\mu$ that jointly respect 
	$\rho$-zCDP.
\end{theorem}

\begin{proof}
	Algorithm~\ref{appendix:alg:hd_mean_est} begins and ends by scaling the data to have 
	empirical mean $0$ and covariance which is L{\"o}wner upper bounded by $I_d$.
	The covariance scaling parameter is chosen independently of the data and the 
	rest of the steps in the algorithm are invariant under location shift. 
	So, our privacy analysis rests on the application of Algorithm~\ref{appendix:alg:hd_mean_step} 
	in lines~\ref{appendix:ln:hd_mean_est-mvm1} and~\ref{appendix:ln:hd_mean_est-mvm2} of Algorithm~\ref{appendix:alg:hd_mean_est}.

	Algorithm~\ref{appendix:alg:hd_mean_step} interacts with the raw data only in line~\ref{appendix:ln:hd_mean_step-gm}, 
	so satisfying DP reduces to correct specification of $\Delta$ (the $\ell_2$ sensitivity of the mean) and application of the Gaussian mechanism.
	The data are projected into $B_2(\tilde{\theta}, r + \gamma_1)$, and so the most a single data point can be changed in $\ell_2$ norm is $2(r+\gamma_1)$. 
	Because neighboring data sets $X,Y$ differ in only one point (call it $z$), the $\ell_2$ norm of the $k-1$ other points remains the same and so 
	\[
		\left\Vert \frac{1}{k} \sum_{x \in X} x - \frac{1}{k}\sum_{y \in Y} y \right\Vert_2 
			= \left\Vert \frac{1}{k} z \right\Vert_2
			= \frac{1}{k} \left\Vert z \right\Vert_2
			\leq \frac{2(r+\gamma_1)}{k}
	\]
 	as desired.	
	Thus, each step of CoinPress satisfies zCDP at the stated level of its privacy parameter $\rho$. 
	For each step $m \in [t-1]$, we see in line~\ref{appendix:ln:hd_mean_est-mvm1} that the privacy parameter is 
	$\tfrac{\rho}{2(t-1)}$.
	For step $t$, we see in line~\ref{appendix:ln:hd_mean_est-mvm2} that the privacy parameter is 
	$\tfrac{\rho}{2}$. Because zCDP parameters compose additively,
	the zCDP parameter for the entire CoinPress algorithm is 
	$(t-1) \tfrac{\rho}{2(t-1)} + \tfrac{\rho}{2} = \rho$.
\end{proof}

\subsection{Proof of Theorem~\ref{theorem:mean_est_form}}
\label{appendix:subsection:proof:mean_est_form}

\begin{proof}
	We start with Assumption~\ref{assump:initial_parameter_bounds} 
	so we have $\mu \in B_2 \left( \tilde{\mu}_0, r_0 \right)$. 
	Note that the clipping bounds, parameterized by $\gamma_1$, in line~\ref{appendix:ln:hd_mean_step-gamma1} of 
	Algorithm~\ref{appendix:alg:hd_mean_step} are set such that, with probability 
	$1 - \beta_s$, no points are affected by the bounding; this follows because any given point is affected 
	only if it falls outside the clipping ball, which occurs with probability $\leq \frac{\beta_s}{k}$ and so,
	by the union bound, every point is unaffected with probability $\geq 1 - \beta_s$. 
	Thus, 
	with probability $\geq 1 - \beta_s$:
	\begin{align*}
		\mu' 
			&\sim  \frac{1}{k} \sum_i^k \hat{\mu}_i + Y \\
			&\sim \hat{\mu} + Y && \left( \text{definition of } \hat{\mu} \right) \\
			&\sim \text{N} \left( \hat{\mu}, \sigma^2 I_d \right).
	\end{align*}

	We now consider $\gamma_2$, which is set as a $1 - \beta_s$ probability upper bound on the $\ell_2$ norm of the 
	privatized mean of $k$ draws from $Q(0, \tilde{\Sigma})$. 
	Conditional on no points being clipped so that $\tilde{\mu}' = \sum_{i=1}^{k} \hat{\mu}_i + Y$, 
	we have
	\begin{align}
		1 - \beta_s
			&\leq \Pr \left( \left\Vert \frac{1}{k} \sum_{i=1}^{k} \hat{\mu}_i - \mu + Y \right\Vert_2 \leq \gamma_2 \right) \\
			&= \Pr \left( \left\Vert \tilde{\mu}' - \mu \right\Vert_2 \leq \gamma_2 \right).
	\end{align}
	So, having $\mu \in B_2(\tilde{\mu}_0, r_0)$ implies that 
	$\Pr \left( \mu \in B_2(\tilde{\mu}', r') \right) \geq 1 - 2\beta_s = 1 - \beta_m$.
	Using the fact that $\sum_m^t \beta_m = \beta^{\mu}$ and a union bound, we proceed by induction over the $t$ steps of the algorithm
	and see that with probability $1 - \beta^{\mu}$ we have  
	\[ \forall m \in [t]: \mu \in B_2(\tilde{\mu}_m, r_m) \]
	and 
	\[ \forall m \in [t]: \mu_m' \sim \text{N} \left( \hat{\mu}, \sigma_m^2 I_d \right). \]
	Scaling $\mu_m', \sigma^2_m$ back up as in Lines~\ref{appendix:ln:hd_mean_est-scale_mu} and~\ref{appendix:ln:hd_mean_est-scale_sigma}
	give the desired result. 
\end{proof}

\subsection{Setting $\gamma_1, \gamma_2$}
\label{appendix:subsection:setting_gamma_1_gamma_2}
This section is concerned with how to set $\gamma_1,\gamma_2$ in 
lines~\ref{appendix:ln:hd_mean_step-gamma1},~\ref{appendix:ln:hd_mean_step-gamma2} of Algorithm~\ref{appendix:alg:hd_mean_step}
for various $Q_{\tilde{\mu}}$.
We start with a general statement that works for arbitrary $Q_{\tilde{\mu}}$.

\begin{fact}[Chebyshev's Inequality]
	If $X$ is a $d$-dimensional random vector with expected value $\mu = \E(X)$ and covariance 
	$\Sigma = \E \left( (X-\mu)(X-\mu)^T \right)$, then 
	\[ \Pr \left( \sqrt{ (X-\mu)^T \Sigma^{-1} (X-\mu) } > t \right) \leq \frac{d}{t^2}, \]
	provided that $\Sigma$ is positive definite.
\end{fact}

\begin{corollary}
	\label{appendix:corollary:general_R_bound}
	For any $R$ in Algorithm~\ref{appendix:alg:hd_mean_step},
	$\Pr \left( \Vert R \Vert_2 > \sqrt{d / \beta} \right) \leq \beta$.
\end{corollary}
\begin{proof}
	By construction of $R$, we know that $\mu = 0$ and $\Sigma = I_d$.
	Let $R_j$ be the $j^{th}$ element of $R$.
	Then we can write 
	\[
		\sqrt{ (R-\mu)^T \Sigma^{-1} (R-\mu) }
			= \sqrt{ R^T R }
			= \left(\sum_{j=1}^{d} R_j^2\right)^{1/2}
			= \Vert R \Vert_2
	\]
	We can set $t = \sqrt{d / \beta}$ and rewrite Chebyshev's Inequality as
	\[
		\Pr \left( \Vert R \Vert_2 > \sqrt{d / \beta} \right) \leq \beta. 
	\] 
\end{proof}

In practice, it is beneficial to set tighter bounds 
based on the specified $Q_{\tilde{\mu}}$.
This can hypothetically be done via Monte Carlo sampling and empirical CDF 
inequalities .
However, this can be computationally expensive for $\gamma_1$ in particular,
as you need at least $k / \beta$ draws (and often far more) from the random variable
to get a proper upper bound. 

Some $Q_{\tilde{\mu}}$ also admit analytical bounds, which avoid the need for the 
costly computation. If $Q_{\tilde{\mu}}$ is multivariate Gaussian, we can use the following: 
\begin{fact}[Lemma 1 of~\cite{LM00}]
	\label{appendix:fact:bounding_R-N}
	Let $Q_{\tilde{\mu}}$ be multivariate Gaussian such that $Q_{\tilde{\mu}}(\mu, \Sigma) = N \left( \mu, \Sigma \right)$.
	Then if $R \sim Q_{\tilde{\mu}}(0, I_d)$, we know that
	\[ \forall \beta \in (0,1]: \Pr \left( \Vert R \Vert_2 > \sqrt{d + \sqrt{d \log(1/\beta)} + 2 \log(1/\beta)} \right) \leq \beta. \]
\end{fact}

We present a similar bound for when $Q_{\tilde{\mu}}$ is multivariate Laplace, 
based heavily on a result from Corollary 3.1 from~\citet{VGNA20}.
\begin{theorem}
	\label{appendix:theorem:bounding_R-MVL}
	Let $Q_{\tilde{\mu}}$ be multivariate Laplace with mean $\mu = 0$ and covariance $\Sigma = I_d$.
	Then if $R \sim Q_{\tilde{\mu}}(0, I_d)$, we know that 
	\[ \forall \beta \in (0,1]: \Pr \left( \Vert R \Vert_2 > \sqrt{e \cdot d\log^2(\beta)} \right) \leq \beta, \]
	where $e \approx 2.718$ is Euler's number.
\end{theorem}
\begin{proof}
	We start by noting that $\Vert R \Vert_2 = \left( \sum_{j=1}^{d} R_j^2 \right)^{1/2}$.
	We know that the $R_j$ are Laplace with mean $0$ and variance $1$,
	and thus $\forall j \in [d]: R_j^2 \sim Weibull(\lambda = 1/2, k = 1/2)$.
	For ease of notation, we'll call $X_j = R_j^2$.

	We now define sub-Weibull random variables, as is done in~\citet{VGNA20}.
	We call a random variable $X_j$ \emph{sub-Weibull} with tail parameter $\theta$ if there exists $\theta,a,b > 0$
	such that 
	\[ \forall x > 0: \Pr \left( \vert X_j \vert \geq x \right) \leq a \exp \left( -b x^{1/\theta} \right). \]
	For context, sub-Gaussian random variables are sub-Weibull with $\theta = 1/2$, sub-Exponentials are 
	sub-Weibull with $\theta = 1$, and Weibull random variables themselves are sub-Weibull with $\theta = 2$.

	We can state an alternative condition, also from~\citet{VGNA20}, 
	that $X_j$ is sub-Weibull with tail parameter $\theta$ if
	\[ \exists c > 0 \text{ s.t. } \forall t \geq 1: \Vert X_j \Vert_t \leq c t^{\theta}. \]

	Our goal is to find the smallest $c$ that holds for Weibull random variables 
	in particular. We recall that $X_j \sim Weibull(\lambda = 1/2, k = 1/2)$ and 
	$\theta = 2$. Thus, for all $t \geq 1$:
	\begin{align} 
		 \Vert X_j \Vert_t &\leq c t^{\theta} \nonumber \\
	\iff \left( \E \left( \left\vert X_j \right\vert^t \right) \right)^{1/t} &\leq c t^{\theta} \nonumber \\
	\iff \lambda \Gamma \left( \frac{t}{k} + 1 \right)^{1/t} &\leq c t^{\theta} \label{appendix:ln:mvl_proof_1} \\
	\iff \frac{1}{2t^2} \Gamma(2t+1)^{1/t} &\leq c \label{appendix:ln:mvl_proof_2}.
	\end{align} 
	Line~\ref{appendix:ln:mvl_proof_1} follows by using the MGF of a Weibull
	random variable, and line~\ref{appendix:ln:mvl_proof_2} follows by plugging in the parameter values.

	Our goal is to find the smallest $c$ such that $\Vert X \Vert_t \leq c t^{\theta}$ 
	for all $t \geq 1$. The lefthand side of line~\ref{appendix:ln:mvl_proof_2} is decreasing in $t$ 
	for $t \geq 1$, so finding the smallest possible $c$ for $t=1$ will be sufficient for all $t \geq 1$.
	Plugging in $t=1$, we get $c=1$.

	We can finally appeal to Corollary 3.1 from~\citet{VGNA20}, which states that 
	if $X_1, \hdots, X_d$ are i.i.d. Weibull random variables with tail parameter 
	$\theta$, then for all $x \geq d K_{\theta}$ we have 
	\[ \Pr \left( \left\vert \sum_{j=1}^{d} X_j \right\vert \geq x \right)
		\leq \exp \left( - \left( \frac{x}{K_{\theta} d} \right)^{1/\theta} \right)
	\] 
	for $K_{\theta} = e c$.
	Plugging in the $c=1$ we found for Weibull random variables yields
	\[ \Pr \left( \left\vert \sum_{j=1}^{d} X_j \right\vert \geq x \right)
		\leq \exp \left( - \left( \frac{x}{e \cdot d} \right)^{1/\theta} \right).
	\] 
	We want the probability to be less than $\beta$, so we sub this in and get 
	\[ \Pr \left( \left\vert \sum_{j=1}^{d} X_j \right\vert \geq e \cdot d \log^2(\beta) \right) \leq \beta. \]
	We note that $\Vert X \Vert_2 = \sqrt{\left\vert \sum_{j=1}^{d} X_j \right\vert}$, so 
	setting the bound at $\sqrt{e \cdot d \log^2(\beta)}$ gives our desired result.
\end{proof}

\subsection{Trick for setting $\tilde{\Sigma}$ for $\hat{\theta}^{BLB}$ estimation}
\label{appendix:subsection:trick_for_mean_estimation}

In our GVDP algorithm, we independently estimate the means of both the 
$\{ \hat{\theta}_i^{BLB} \}_{i \in [k]}$ and $\{ \hat{\Sigma}_i^{BLB} \}_{i \in [k]}$, each requiring (among other things)
that the analyst specify $\tilde{\Sigma}$, a L{\"o}wner upper bound on the sample 
covariance of the BLB samples. If we estimate the mean of $\{ \hat{\Sigma}_i^{BLB} \}_{i \in [k]}$ and do our post-processing 
to find a private covariance estimate $\tilde{\Sigma}$
prior to estimating the mean of $\{ \hat{\theta}_i^{BLB} \}_{i \in [k]}$, we can actually leverage some extra information 
that will generally improve our estimates with a small cost to the theoretical guarantee.
All the experimental results in the paper use this trick.

Although we are scaling up our subsets to the original data size within the BLB to get correct overall covariance 
estimates, this does not imply that the covariance of the $\{ \hat{\theta}_i^{BLB} \}_{i \in [k]}$ match this
correct scaling. In fact, this covariance will often be roughly the same as if $\hat{\theta}$ were simply 
run on subsets of size $\frac{n}{k}$. So, the covariance of the $\{ \hat{\theta}_i^{BLB} \}_{i \in [k]}$
should be roughly $\frac{r(n/k)}{r(n)} \hat{\Sigma}$, where $r$ is the convergence rate of the estimator 
in question. For example, the covariance of OLS coefficients decays with $\frac{1}{n}$, so if $\hat{\theta}$
represents OLS estimation we would say the covariance is $\frac{1/(n/k)}{1/n}\hat{\Sigma} = k \hat{\Sigma}$. 
We upper bound this with $k \tilde{\Sigma}$.

Under Assumption~\ref{assump:blb_approx}, this strategy gives us
a $1-\beta^{\tilde{\Sigma}}$ probability guarantee that $k \tilde{\Sigma}$ will 
L{\"o}wner upper bound the empirical covariance of the $\{ \hat{\theta}_i^{BLB} \}_{i \in [k]}$. 
So, by using $k \tilde{\Sigma}$ as the upper covariance bound for our mean estimation for 
$\{ \hat{\theta}_i^{BLB} \}_{i \in [k]}$, we generally start with a pretty tight bound 
and can dramatically improve the accuracy of our estimates. 
This does lose a bit of theoretical strength in the results; we generally assume that 
the analyst's upper bound is an actual upper bound on the empirical covariance with probability 1,
whereas this trick provides a guarantee with probability $1-\beta^{\tilde{\Sigma}}$.

\subsection{Generalizing CoinPress beyond multivariate sub-Gaussians}
In Figure~\ref{appendix:fig:multivariate_laplace} we provide evidence that our generalization of CoinPress beyond sub-Gaussian distributions delivers on its promises.
We pretend as if the estimator and data were such that the distributions induced by the BLB were 
are dominated by the multivariate Laplace (i.e. they are sub-Exponential),
and the analyst overestimated the relevant parameters by a factor of 100. 
We show results corresponding to two different methods for calculating the clipping parameters at each step of CoinPress.
The \emph{analytic} solution calculates the bound using a theoretical bound given in Theorem~\ref{appendix:theorem:bounding_R-MVL},
while the \emph{approximate} solution calculates an approximate upper bound using Monte Carlo sampling. 

\begin{figure}
    \captionsetup{font=footnotesize,labelfont=footnotesize}
    \begin{subfigure}[t]{.49\textwidth}
        \centering
        \includegraphics[width=\textwidth]{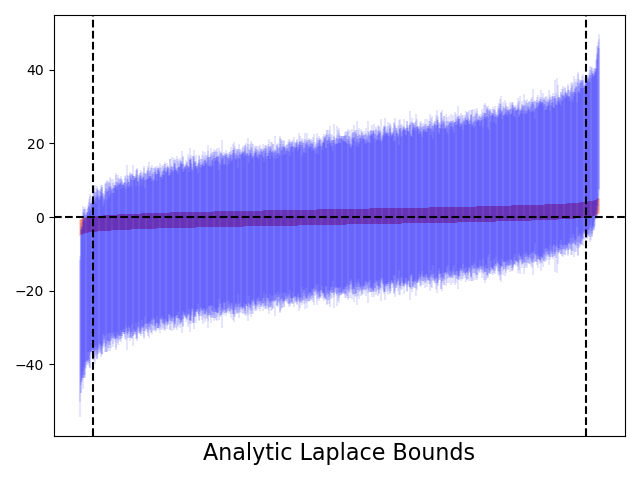}
    \end{subfigure}
    \begin{subfigure}[t]{.49\textwidth}
        \centering
        \includegraphics[width=\textwidth]{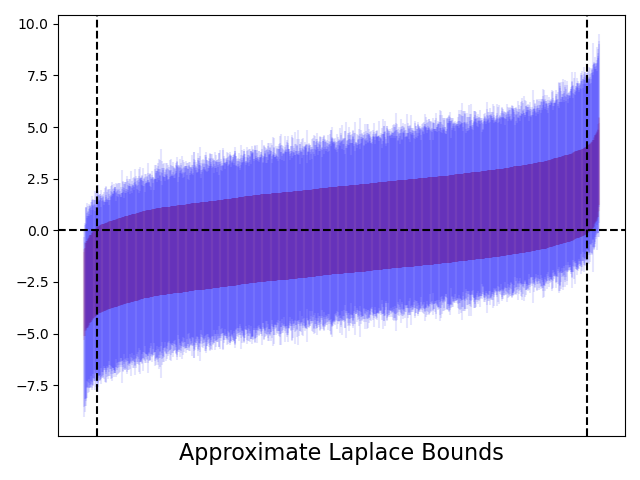}
    \end{subfigure}
    \caption{Distribution of coefficient estimates and 95\% confidence intervals for  
                $k = 5{,}000, d = 10, \rho = 0.1$ for multivariate Laplace distribution}
    \label{appendix:fig:multivariate_laplace}
\end{figure}

\section{Step 3: Postprocessing}
\label{appendix:section:parameter_estimation}

\subsection{Proof of Theorem~\ref{theorem:precision_weighting}}
\label{appendix:subsection:proof:precision_weighting}
\begin{proof}
	Our goal is to find weights $\{A_m\}_{m \in [t]}$ with $A_m \in \R^{d \times d}$ such that the 
	L{\"o}wner order of $\text{Cov} \left( \sum_{m=1}^{t} A_m \hat{\tau}_m \right)$ is minimized.
	Because we want our weighted estimator to remain unbiased, we restrict ourselves to 
	sets of $A_m$ such that $\sum_{m=1}^{t}A_m = I_d$.

	We note that the $A_m$ are constants and $\hat{\tau}_m$ are independent, so 
	\begin{align*}
		\text{Cov}\left( \sum_{m=1}^{t} A_m \hat{\tau}_m \right)
			&= \sum_{m=1}^{t} \text{Cov} \left( A_m \hat{\tau}_m \right) \\
			&= \sum_{m=1}^{t} A_m^T \text{Cov} \left( \hat{\tau}_m \right) A_m.
	\end{align*}

	Assume $\hat{\tau}_m \in \R^{d}$ and let $\{B_m\}_{m \in [t]}$ with $B_m \in \R^{d \times d}$ be 
	an arbitrary weighting. Then we can write 
	\begin{align*} 
		\text{Cov}\left( \sum_{m=1}^{t} A_m \hat{\tau}_m \right) &\preceq \text{Cov}\left( \sum_{m=1}^{t} B_m \hat{\tau}_m \right) \\
		\iff \forall v \in \R^d \setminus \{0\}: v^T \text{Cov}\left( \sum_{m=1}^{t} A_m \hat{\tau}_m \right) v &\leq v^T \text{Cov}\left( \sum_{m=1}^{t} B_m \hat{\tau}_m \right) v.
	\end{align*}
	Note that the quantities on the righthand side of the statement above are scalars, 
	so we have translated the problem of finding a minimal L{\"o}wner bound into minimizing a one-dimensional 
	quantity. 

	Let $v \in \R^d \setminus \{0\}$ be arbitrary.
	We now have a one-dimensional constrained optimization problem; 
	we want to find $\{A_m\}_{m \in [t]}$ which minimizes 
	$v^T \text{Cov}\left( \sum_{m=1}^{t} A_m \hat{\tau}_m \right) v$ 
	subject to $\sum_{m=1}^{t}A_m = I_d$.
	We can solve this using a Lagrange multiplier.

	We write
	\[ \mathcal{L}\left( \{A_m\}_{m \in [t]}, \lambda \right) = 
		v^T \text{Cov}\left( \sum_{m=1}^{t} A_m \hat{\tau}_m \right) v - \lambda v^T \left( \sum_{m=1}^{t} A_m - I_d \right) v 
	\]
	and differentiate with respect to $A_m$. Recall that $\text{Cov} \left( \hat{\tau}_m \right) = S_m$.
	Then we have
	\begin{align}
		\frac{\partial \mathcal{L}\left( \{A_m\}_{m \in [t]}, \lambda \right)}{\partial A_m}
			&= \frac{ \partial v^T \text{Cov}\left( \sum_{m=1}^{t} A_m \hat{\tau}_m \right) v - \lambda v^T \left( \sum_{m=1}^{t} A_m - I_d \right) v }{ \partial A_m } \nonumber \\
			&= \frac{ \partial \left( \sum_{m=1}^{t} v^T A_m^T \text{Cov} \left( \hat{\tau}_m \right) A_m v \right) - \lambda v^T \left( \sum_{m=1}^{t} A_m - I_d \right) v }{ \partial A_m } \nonumber \\
			&= S_m A_m v v^T + S_m^T A_m v v^T - \lambda v v^T \label{appendix:ln:matrix_calc_identity_1} \\
			&= 2 \left( S_m A_m - \lambda I_d \right) v v^T. \nonumber
	\end{align}
	(\ref{appendix:ln:matrix_calc_identity_1}) comes from a matrix calculus identity that for vectors $a,b$ and matrix $C$ all independent of $X$,
	$\frac{\partial (Xa)^T C (Xb)}{\partial X} = C X b a^T + C^T X a b^T$ and noting that the partial 
	with respect to $A_m$ influences the sum only in the $m^{th}$ term.

	We set this to $0$ to find a stationary point. 
	\begin{align}
		0 &= 2 \left( S_m A_m - \lambda I_d \right) v v^t \nonumber \\
		\lambda I_d v v^T &= S_m A_m v v^t \nonumber \\
		A_m &= \lambda S_m^{-1} I_d v v^T (v v^T)^{-1} \nonumber \\
			&= \lambda S_m^{-1}. \nonumber
	\end{align}
	We know from our constraint that $\sum_{m=1}^{t} A_m = I_d$, so 
	\begin{align}
		\sum_{m=1}^{t} \lambda S_m^{-1} &= I_d \nonumber \\
		\lambda &= \left( \sum_{m=1}^{t} S_m^{-1} \right)^{-1}, \nonumber
	\end{align}
	and thus our stationary point is achieved at $A_m = \left( \sum_{m=1}^{t} S_m^{-1} \right)^{-1} S_m^{-1}$.

	We have shown that choosing $A_m$ in this way achieves a stationary point, but we want to show that it is a global minimum.
	For that, we need to check the second partial derivative test, which states that our stationary point is a global minumum 
	if $\frac{ \partial^2 \mathcal{L}\left( \{A_m\}_{m \in [t]}, \lambda \right) }{\partial^2 A_m}$ is PD.

	We first note that 
	\begin{align}
		\frac{ \partial^2 \mathcal{L}\left( \{A_m\}_{m \in [t]}, \lambda \right) }{\partial^2 A_m} \nonumber 
			&= \frac{\partial}{\partial A_m} 2 \left( S_m A_m - \lambda I_d \right) v v^T \nonumber \\
			&= 2 (v v^T) \otimes S_m, \nonumber
	\end{align}
	where $\otimes$ is the Kronecker product.

	We know $v v^T$ is PD, because $\forall z \in \R^d \setminus \{0\}$ we get 
	$z^T v v^T z = (z^T v) (v^T z) = (v^T z)^T (v^T z) > 0$. 
	The strict inequality comes because we know that both $v$ and $z$ are non-zero. 
	We know $S_m$ is PD by assumption and that, in general, if a matrix $Y$ is PD then 
	so is $2Y$. Finally, the Kronecker product of PD matrices is also PD, so $2 (v v^T) \otimes S_m$ is PD 
	and our second partial derivative condition is met. So $\mathcal{L}$ is convex and our local minimum 
	is also a global minimum.
	Thus, our choice of $A_m$ 
	achieves the $\text{Cov}\left( \sum_{m=1}^{t} A_m \hat{\tau}_m \right)$ with minimal L{\"o}wner order.
\end{proof}

\subsection{Proof of Theorem~\ref{theorem:tilde_Sigma}}
\label{appendix:subsection:proof:tilde_Sigma}

\begin{proof}

	From Theorem~\ref{theorem:mean_est_form}, 
	we know that, with probability $\geq 1 - \beta^{\tilde{\Sigma}}$: 
	\[
		\forall m \in [t]: 
		\tilde{S}_m \sim \text{N} \left( \hat{S}^{BLB}, \vec{\sigma}^2_{\tilde{\Sigma}, m} I_{d'} \right),
	\]
	where $\hat{S}^{BLB}$ is the flattened form of $\hat{\Sigma}^{BLB}$.
	Assumption~\ref{assump:blb_approx} then let's us substitute in $\hat{S}$ for $\hat{S}^{BLB}$.
	For the rest of the proof, we assume 
	that this condition is met.

	Thus, by Theorem~\ref{theorem:precision_weighting}, we know that
	a precision-weighted $\tilde{S}$ will have mean $\E[\tilde{S}] = \hat{S}$ and covariance
	$\text{Cov} \left( \tilde{S} \right) = \left( \Sigma_{m=1}^{t} \vec{\sigma}^{-2}_{\tilde{\Sigma}, m} I_{d'} \right)^{-1}$
	Moreover, this $\tilde{S}$ is itself multivariate Gaussian because it is a linear combination 
	of multivariate Gaussians.
	That is, we can write 
	\[ 
		\tilde{S} 
			\sim \text{N} \left( \hat{S}, \left( \sum_{m=1}^{t} \vec{\sigma}^2_{\tilde{\Sigma}, m} I_{d'} \right)^{-1} \right)
			=:
			\text{N} \left( \hat{S}, \vec{\sigma}^2_{\tilde{S}} I_{d'} \right).
	\]
	Let $\tilde{\Sigma}'$ be the unflattened matrix constructed from $\tilde{S}$.
	Then we can write $\tilde{\Sigma}'_{i,j} \sim N \left( \hat{\Sigma}_{i,j}, b_{i,j}^2 \right)$,
	where $b_{i,j} = \text{unflatten} \left( \vec{\sigma}_{\tilde{S}} \right)_{i,j}$.
	Then, by Theorem 1.1 from~\citet{BVH16}, we know that 
	\begin{equation*}
		\E \norm{ \tilde{\Sigma}' - \hat{\Sigma} }_2
			\leq (1+\epsilon) \left( 2 \max_{i \in [d]} \norm{b_{i, \cdot}}_2 + \frac{6 \sqrt{\log d}}{\log(1+\epsilon)} \max_{i,j \in [d] \times [d]} |b_{i,j}| \right),
	\end{equation*}
	for arbitrary $\epsilon \in (0, 1/2]$, where $\norm{\cdot}_2$ is the spectral norm.
	Moreover, by Corollary 3.9 from~\citet{BVH16} we have that, for any $t \geq 0$:
	\begin{equation*}
		\norm{ \tilde{\Sigma}' - \hat{\Sigma} }_2
			\leq (1+\epsilon) \left( 2 \max_{i \in [d]} \norm{b_{i, \cdot}}_2 + \frac{6 \sqrt{\log d}}{\log(1+\epsilon)} \max_{i,j \in [d] \times [d]} |b_{i,j}| \right) + t
	\end{equation*}
	with probability $\geq 1 - \exp \left( \frac{-t^2}{4 \max_{i,j} b_{i,j}^2}\right)$.
	Setting $t = \sqrt{ \frac{\ln(1 / \beta^{ub})}{ 4 \max_{i,j} b_{i,j}^2 } }$
	yields a $1-\beta^{ub}$ probability bound.

 	Now define 
	\[ 
		c = \min_{\epsilon \in (0, 1/2]} (1+\epsilon) \left( 2 \max_{i \in [d]} \norm{b_{i, \cdot}}_2 + \frac{6 \sqrt{\log d}}{\log(1+\epsilon)} \max_{i,j \in [d] \times [d]} |b_{i,j}| \right) + \sqrt{ \frac{\ln(1 / \beta^{ub})}{ 4 \max_{i,j} b_{i,j}^2 } }.
	\]
	The spectral norm $\norm{\cdot}_2$ of a matrix is its largest singular value (or equivalently, the square root 
	of the absolute value of its largest magnitude eigenvalue). So, if 
	$\norm{ \tilde{\Sigma}' - \hat{\Sigma} }_2 \leq c$, we know that the smallest eigenvalue of $\tilde{\Sigma}' - \hat{\Sigma}$ 
	is necessarily at least $-c$.
	Therefore, the smallest eigenvalue of $\tilde{\Sigma}' + cI_{d} - \hat{\Sigma}$ is at least $0$ or, equivalently,
	$\tilde{\Sigma} + c I_d \succeq \hat{\Sigma}$. This statement holds with probability $1 - \beta^{ub}$. Combining this 
	with the initial $1 - \beta^{\tilde{\Sigma}}$ probability guarantee on the form of our estimator completes the proof.
\end{proof}

The statement for $c$ simplifies significantly in the case where the $\{ \tilde{\Sigma}_m \}_{m \in [t]}$ are diagonal matrices 
(which occurs if we care only about confidence intervals for each parameter rather than a joint confidence region).
Let $q(p, \mu, \sigma^2) := \sqrt{2}\sigma \text{erf}^{-1}(2p-1) + \mu$ be the quantile function for a $N(\mu,\sigma^2)$ distribution 
where $\text{erf}^{-1}(\cdot)$ is the inverse error function.  

\begin{corollary}
	\label{theorem:tilde_Sigma_diagonal}
	Given diagonal covariance estimates and privacy variances $\{ \tilde{\Sigma}_m, \vec{\sigma}^2_{\Sigma, m} \}_{m \in [t]}$,
    let $\tilde{S}_m \in \R^{d'}$ be the flattened version of $\tilde{\Sigma}_m$.
    We can construct a precision-weighted estimator:
	$
		\tilde{S} := \frac{ \sum_{m=1}^{t} \tilde{S}_m / \vec{\sigma}_{\Sigma,m}^2 }{ \sum_{m=1}^{t} 1 / \vec{\sigma}_{\Sigma,m}^2 }.
    $

    Let $\tilde{\Sigma}'$ be the diagonal $d \times d$ matrix created by unflattening $\tilde{S}$ and $b$ be the unflattened $d \times d$ diagonal
    matrix where $b^2_{i,i} = \Var \left( \tilde{\Sigma}'_{i,i} \right)$ (i.e.\ the diagonal values of the covariance matrix of the flattened 
    precision-weighted estimator). For $\beta^{ub} \in (0,1)$, define $\vec{c} = \{ c_j \}_{j \in [d]}$ where
	\[
		c_j = q \left( 1 - \frac{\beta^{ub}}{d}, 0, b_{j,j}^2 \right).
	\]

	Then, for 
    $\tilde{\Sigma} = \tilde{\Sigma}' + \vec{c} I_{d}$
	we have
	$
        \Pr \left( \forall j \in [d]: \hat{\Sigma}_{j,j} \leq \tilde{\Sigma}_{j,j} \right) \geq 1 - \beta^{\tilde{\Sigma}} - \beta^{ub}. 
    $
\end{corollary}

\begin{proof}
	We start as in the proof in Section~\ref{appendix:subsection:proof:tilde_Sigma}, but we know additionally that $\tilde{\Sigma}'_{i,j} = 0$
	for $i \neq j$.
	We know that if our assumptions hold, which happens with probability $\geq 1 - \beta^{\tilde{\Sigma}}$, 
	we can write $\tilde{\Sigma}'_{j,j} \sim N \left( \hat{\Sigma}_{j,j}, b_{j,j}^2 \right)$
	where $b_{j,j} = \left( \vec{\sigma}_{\tilde{S}} \right)_j$.

	By definition of the quantile function, we know then that, for arbitrary $j \in [d]$:
	\begin{align*}
		1 - \frac{\beta^{ub}}{d} 
			&= \Pr \left( \tilde{\Sigma}'_{j,j} \leq q \left(1-\frac{\beta^{ub}}{d}, \hat{\Sigma}_{j,j}, b_{j,j}^2 \right) \right) \\
			&= \Pr \left( \tilde{\Sigma}'_{j,j} \leq \hat{\Sigma}_{j,j} + q \left(1-\frac{\beta^{ub}}{d}, 0, b_{j,j}^2 \right) \right) \\
			&= \Pr \left( \tilde{\Sigma}'_{j,j} - c_j \leq \hat{\Sigma}_{j,j} \right) && \left( \text{definition of } c_j \right) \\
			&= \Pr \left( \tilde{\Sigma}'_{j,j} + c_j \geq \hat{\Sigma}_{j,j} \right) && \left( \text{symmetry of the Gaussian} \right).
	\end{align*}

	Applying a union bound over the $d$ failure probabilities and combining with the $1 - \beta^{\tilde{\Sigma}}$ probability that our required
	assumptions hold yields the desired result.
\end{proof}

\subsection{Proof of Theorem~\ref{thm:tilde_theta}}
\label{appendix:subsection:proof:tilde_theta}

\begin{proof}
	From Theorem~\ref{theorem:mean_est_form}, 
	we know that, with probability $\geq 1 - \beta^{\tilde{\theta}}$: 
	\[
		\forall m \in [t]: 
		\tilde{\theta}_m \sim \text{N} \left( \hat{\theta}^{BLB}, \vec{\theta}^2_{\tilde{\theta}, m} I_{d'} \right),
	\]
	where $\hat{S}$ is the flattened form of $\hat{\Sigma}$.
	Assumption~\ref{assump:blb_approx} then let's us substitute in $\hat{\theta}$ for $\hat{\theta}^{BLB}$.

	However, we opt to use the trick from Section~\ref{appendix:subsection:trick_for_mean_estimation} to avoid having to make 
	Assumption~\ref{assump:initial_covariance_bounds}. Theorem~\ref{thm:tilde_theta} gives us a covariance bound 
	that, after appropriate scaling, satisfies Assumption~\ref{assump:initial_covariance_bounds} with probability 
	$1 - \beta^{\tilde{\Sigma}} - \beta^{ub}$, so we fold this into our failure probability.
	Our result then follows directly from the precision-weighting procedure in Theorem~\ref{theorem:precision_weighting}.
\end{proof}

\section{Confidence Region/Intervals}
\label{appendix:section:confidence_intervals}

\subsection{Proof of Theorem~\ref{theorem:general_cr_high_prob}}
\label{appendix:subsection:general_cr_high_prob}

\begin{proof}
	We assume that our estimation of $\tilde{\theta}$ and $\tilde{\Sigma}$ worked 
	as described at the top of Section~\ref{subsubsection:confidence_intervals}, which 
	comes with a $1 - \beta^{\tilde{\Sigma}} - \beta^{ub} - \beta^{\tilde{\theta}}$
	probability guarantee. This means that $\E(\tilde{\theta}) = \E(\hat{\theta}) = \theta$
	and $\tilde{\Sigma} \succeq \hat{\Sigma} \succeq \Sigma$ where our non-private estimator
	$\hat{\theta} \sim G(\theta, \Sigma)$.
	Furthermore, we assume that $Q_{\hat{\theta}}$ is heavier-tailed than $G$. 
	
	Summarizing this, we have a high-probability guarantee that three conditions hold: \\

	\noindent \textbf{(1)} $\E(Z) = \E(\hat{\theta}) = \theta$ \\
	\noindent \textbf{(2)} $\Sigma \preceq \tilde{\Sigma}$ \\
	\noindent \textbf{(3)} $Q_{\hat{\theta}}$ is heavier-tailed than $G$.

	Under these conditions, $Z$ and $\hat{\theta}$ have the same mean and $\hat{\theta}$ is more concentrated than $Z$, 
	so a confidence region that is valid for $Z$ is valid for $\hat{\theta}$. 
	In other words,
	\[ \forall \alpha \in (0,1): \Pr(Z \in C) \geq 1 - \alpha \implies \Pr(\theta \in C) \geq 1 - \alpha. \]
\end{proof}

The result for confidence intervals, rather than a single region, follows trivially by noting that a confidence interval is a 
1-dimensional confidence region.

\begin{corollary}[Confidence Intervals (valid with high probability)]
    \label{theorem:general_ci_high_prob}
    Let $Z$ be a $d$-dimensional random variable such that 
    $Z \sim Q_{\hat{\theta}} \left( \hat{\theta} + \text{N} \left(0, \Sigma_{\tilde{\theta}} \right), \tilde{\Sigma} \right)$.
    Suppose $\{ (ci^l_j, ci^u_j) \}_{j \in [d]}$ is a set of intervals such that 
    \[
        \forall j \in [d]: \Pr \left( Z_j \in (ci^l_j, ci^u_j) \right) \geq 1-\alpha_j,
    \]
    for some $\{\alpha_j\}_{j \in [d]}$ with $\alpha_j \in (0,1)$.
    Then, with probability $1 - \beta^{\tilde{\Sigma}} - \beta^{ub} - \beta^{\tilde{\theta}}$, 
    \[ \forall j \in [d]: \Pr \left( \hat{\theta}^{BLB}_j \in \left( ci^l_j, ci^u_j \right) \right) \geq 1 - \alpha_j. \]

	Likewise, with probability $1 - \beta^{\tilde{\Sigma}} - \beta^{ub} - \beta^{\tilde{\theta}}$, 
    \[\Pr \left( \forall j \in [d] \hat{\theta}^{BLB}_j \in \left( ci^l_j, ci^u_j \right) \right) \geq 1 - \sum_j^d \alpha_j. \]
\end{corollary}

\begin{proof}
	All but the last statement is a trivial application of Theorem~\ref{theorem:general_cr_high_prob} to the 1-dimensional case. 
	The last statement follows via a union bound.
\end{proof}

\section{Empirical Results}
\label{appendix:empirical_results}

\subsection{Logistic regression with imbalanced class output}
\label{appendix:section:logistic_reg_imbalanced}

In Figure~\ref{fig:logistic_regression} we show qualitatively similar results 
for a more challenging setting; logistic regression with imbalanced classes. 

We generate data as we did for Figure~\ref{fig:ols_regression} but run the outcome variable $y$ through a scaled logistic function to get 
a new outcome variable $y' \in \{0,1\}^n$.
Specifically, for $p_i = \frac{1}{1 + \exp(-X_i \beta)}$ and $\bar{p} = \frac{1}{n} \sum_{i=1}^{n}p_i$, we have 
$\Pr(y_i = 1) = \frac{p_i}{\bar{p}} \cdot 0.05$. This induces a minority class that occurs with probability 
$\approx 0.05$. Having such imbalanced classes introduces a practical problem in choosing 
a good $k$ for GVDP. If $\frac{n}{k}$ is small, it becomes likely that we will see 
only the majority class in any given subset and the model will not be able to be fit.
Thus, our experiments here have larger $n$ than we used for the OLS experiments.

\begin{figure}[h!]
    \captionsetup{font=footnotesize,labelfont=footnotesize}
    \begin{subfigure}[t]{.32\textwidth}
        \centering
        \includegraphics[width=\textwidth,height=4cm]{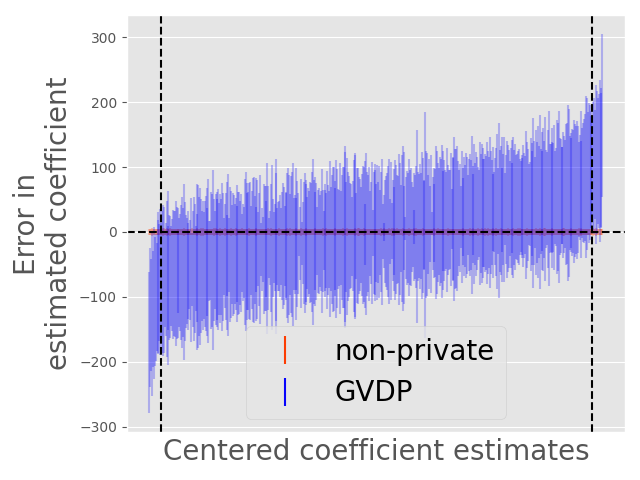}
        \caption{$n = 10{,}000, k = 250$}
    \end{subfigure}
    \begin{subfigure}[t]{.32\textwidth}
        \centering
        \includegraphics[width=\textwidth,height=4cm]{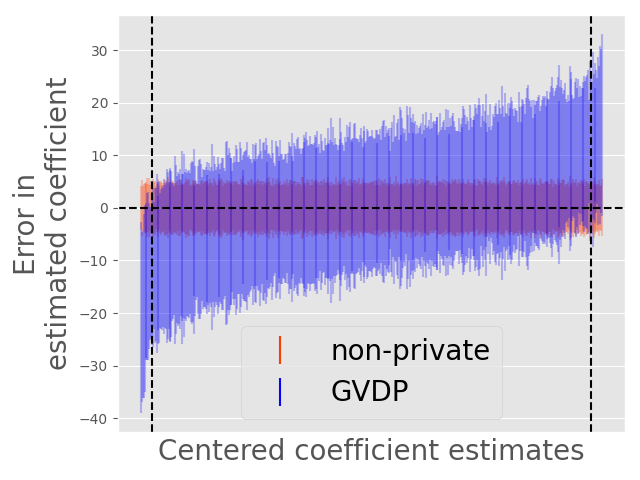}
        \caption{$n = 100{,}000, k = 1{,}000$}
    \end{subfigure}
    \begin{subfigure}[t]{.32\textwidth}
        \centering
        \includegraphics[width=\textwidth,height=4cm]{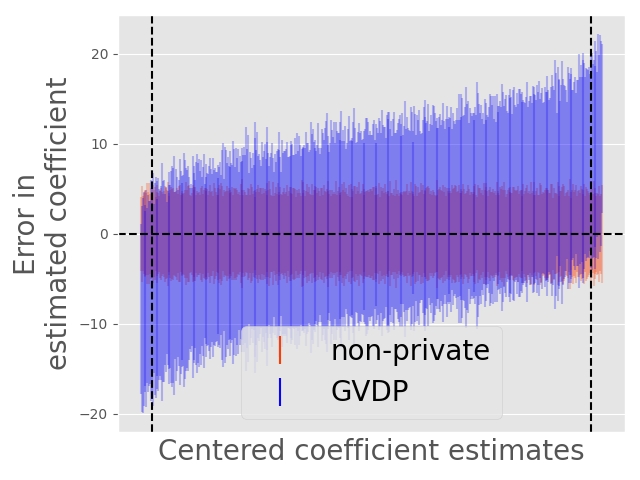}
        \caption{$n = 100{,}000, k = 2{,}000$}
    \end{subfigure}
    \caption{Logistic Regression: Distribution of coefficient estimates and 95\% confidence intervals}
    \label{fig:logistic_regression}
\end{figure} 

\subsection{Logistic regression with fully sparse data}
\label{appendix:section:logistic_reg_imbalanced_binary}

The requisite bootstrap assumptions do not hold for all estimators and data distributions. 
We make our setting more difficult again for Figure~\ref{fig:logistic_regression_binary}, 
by making both the outcome and covariates a sparse binary vector/matrix respectively.
It is unlikely that an analyst would want to use GVDP in this setting, because knowing that the data are binary (which we assume an 
analyst would know)
immediately provides tight clipping bounds for the data. Nevertheless, we include it as an example because 
it's the most natural scenario we found where our method fails because of poor performance of the BLB. 

We create a new set of covariates $X'$
such that $\forall j \in [d]: X'_{i,j} = \mathbbm{1}(X_{i,j} \geq z_j)$ where 
$z_j = \min_{r \in \R} \frac{1}{n}\sum_{i=1}^{n} \mathbbm{1}(X_{i,j} > r) \leq 0.05$.
That is, $X'$ is itself now a binary matrix with highly imbalanced classes.
The rightmost plot shows distributions of the $d$ coefficient estimates induced by BLB, with a black dotted line 
at the value of the non-private coefficient. Note that at $n = 100{,}000$, the BLB distributions 
are essentially point masses at two extreme points, and our algorithm yields biased estimates and confidence intervals with insufficient coverage. 
For $n = 1{,}000{,}000$, the BLB distributions are much closer to being a symmetric 
distribution about the true coefficient value, and our algorithm yields the promised guarantees. 

The left plots show poor confidence interval 
coverage because the BLB distribution is not a good approximation of the non-private sampling distribution.
The right plots show better confidence interval coverage because the BLB approximation is successful.

\begin{figure}[h!]
    \captionsetup{font=footnotesize,labelfont=footnotesize}
    \begin{subfigure}[t]{.49\textwidth}
        \centering
        \includegraphics[width=\textwidth,height=4cm]{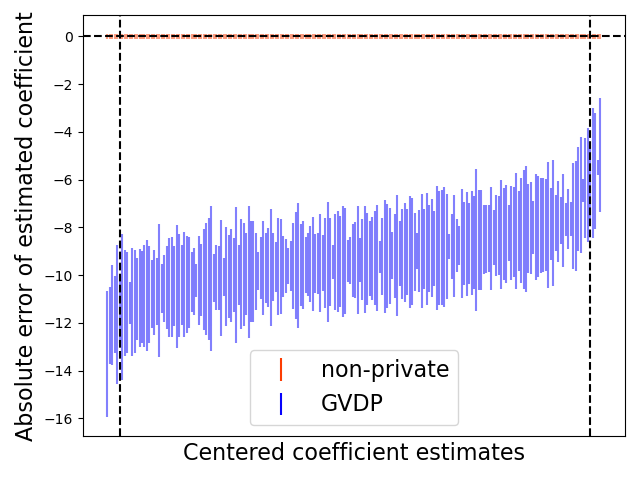}
        \caption{$n = 100{,}000, k = 500, d = 10$}
    \end{subfigure}
    \begin{subfigure}[t]{.49\textwidth}
        \centering
        \includegraphics[width=\textwidth,height=4cm]{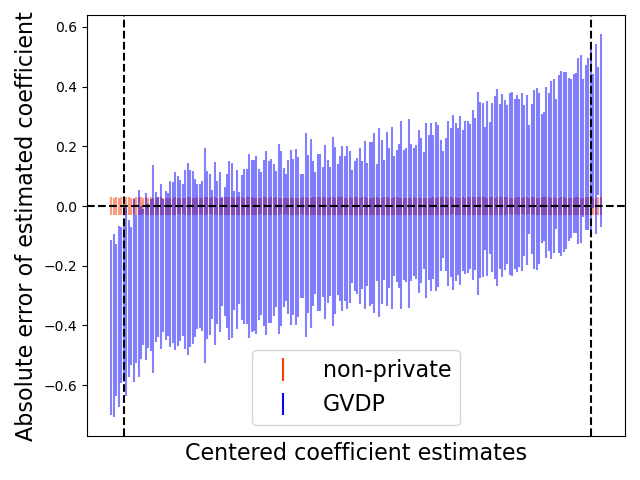}
        \caption{$n = 1{,}000{,}000, k = 500, d = 10$}
    \end{subfigure}
    \begin{subfigure}[t]{.49\textwidth}
        \centering
        \includegraphics[width=\textwidth,height=4cm]{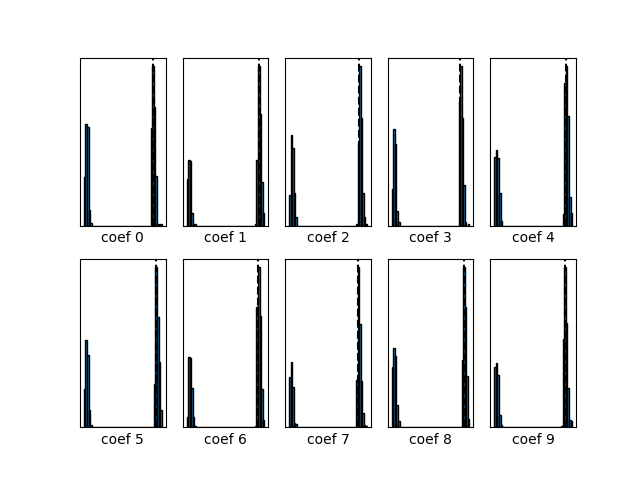}
        \caption{BLB coefficient distributions for\\$n = 100{,}000, k = 500$}    
    \end{subfigure}
    \begin{subfigure}[t]{.49\textwidth}
        \centering
        \includegraphics[width=\textwidth,height=4cm]{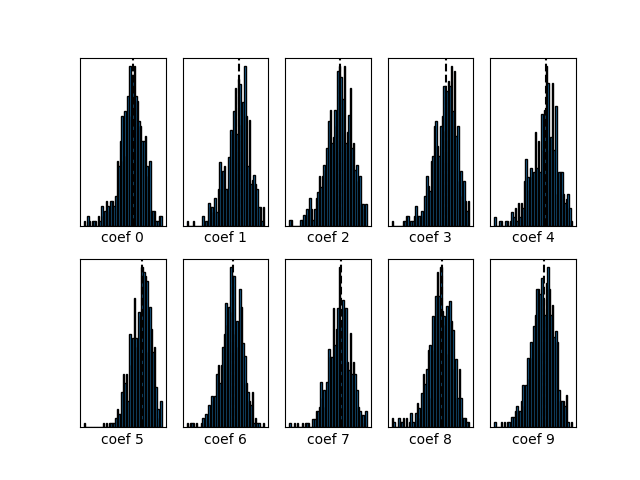}
        \caption{BLB coefficient distributions for\\$n = 1{,}000{,}000, k = 500$}
    \end{subfigure}
    \caption{Logistic Regression with unbalanced binary fatures: Distribution of coefficient estimates and 95\% confidence intervals}
    \label{fig:logistic_regression_binary}
\end{figure} 

\subsection{Explanation of Table~\ref{tab:adassp_comparison}}
\label{appendix:section:explanation_adassp_comparison}

The GVDP and AdaSSP algorithms differ in a few key ways. 
First, AdaSSP does not attempt to do unbiased parameter estimation or give valid confidence intervals; instead, 
it is trying to estimate OLS coefficients with minimal $\ell_2$ error.
For purposes of comparison, we will ignore confidence intervals altogether and focus only 
on the parameter estimates. 
Second, AdaSSP assumes only bounds on the data, assuming that we can specify data domains 
$\mathcal{X}, \mathcal{Y}$ for our covariates and outcome, respectively,
such that $\Vert \mathcal{X} \Vert = \sup_{x \in \mathcal{X}} \Vert x \Vert_2$ where $x \in \R^m$ 
and $\Vert \mathcal{Y} \Vert = \sup_{y \in \mathcal{Y}} \vert y \vert$.  
Ignoring the assumptions needed for confidence intervals for now, 
GVDP requires Assumptions~\ref{assump:initial_parameter_bounds}, and 
\ref{assump:initial_covariance_bounds} on the
distribution of covariances induced by the bag of little bootstraps, as well as 
Assumption~\ref{assump:initial_parameter_bounds} on the distribution of means.
It's worth noting the qualitative difference between these methods; AdaSSP requires the user to bound the data, 
GVDP requires the user to bound moments of the parameter distribution. For most analyses, we expect the 
AdaSSP bounds to be easier to specify tightly than those of GVDP. However, GVDP is designed to scale more gracefully 
under overly conservative bounds. 

We generate data just as we did for our OLS demonstration and 
compare AdaSSP and GVDP across a number of what we call ``overestimation factors''.
Say we have realized data $X \in \R^{n \times d}, y \in \R^d$. 
For an overestimation factor of $c$, we set the bounds for AdaSSP to 
$c \cdot \sup_{x \in X} \Vert x \Vert_2$ and $c \cdot \sup_{y \in Y} \Vert y \Vert$.
For GVDP, we perform the BLB step to get our $\{ \hat{\theta}_i^{BLB} \}_{i \in [k]}$,
which we'll say has empirical mean $\hat{\mu} \in \R^d$ and empirical covariance $\hat{\Sigma} \in \R^{d \times d}$. 
We set our $\ell_2$ bounding ball for the mean of the distribution as 
$B_2 \left( \hat{\mu}, c \left(\max_{j \in [d]}\hat{\mu}_j\right) \right)$ and 
our L{\"o}wner upper bound on the covariance as 
$c \left( \text{diag}(\hat{\Sigma}) I_d \right)$. Runs of non-private OLS are included for comparison, 
but the overestimation factor does not affect them.

The experiment in the body of the paper was run with $n = 500{,}000, k = 2{,}500, d = 10$, and $\rho = 0.1$.
We run each method over 100 simulations, estimating $d$ coefficients at each iteration,
so each method produces $1{,}000$ coefficient estimates overall.

\paragraph{Comparison with AdaSSP}
\label{subsection:comparison_with_adassp}
We again consider OLS, but now compare GVDP's performance to that of the Adaptive Sufficient Statistic Perturbation (AdaSSP) algorithm 
from~\citet{Wan18}, one of the best-performing algorithms for DP OLS. 
AdaSSP assumes bounds on the underlying data and attempts to estimate the OLS coefficients with minimal $\ell_2$ error.
We consider the performance of AdaSSP vs. GVDP as a function of the ``overestimation factor'' (OF), which is a multiplicative 
factor by which we overestimate the bounds of the data (for AdaSSP) or parameters (for GVDP).

\begin{table}
    \captionsetup{font=footnotesize,labelfont=footnotesize}
    \scriptsize
    \centering
    \begin{tabular}{lrrrrrrrr}
        \toprule
        OF      &    1     &  1.5          &  4       &  5       &   10      &   100    &  1000 & 10000 \\
        \midrule
        non-private & 39.23 & 39.23 & 39.23 & 39.23 & 39.23 & 39.23 & 39.23 & 39.23 \\
        AdaSSP & 40.31 & 49.66 & 1801.04 & 23395.30 & 71790.88 & 46382.71 & 82803.01 & 76377.38 \\
        GVDP & 58.23 & 59.01  & 58.64 & 57.74 & 58.61 & 61.57 & 70.24 & 102.05 \\
        \bottomrule
    \end{tabular}
    \caption{Average $\ell_2$ estimation error for each algorithm by overestimation factor (OF)}
    \label{tab:adassp_comparison}
\end{table}

Despite our presentation above, AdaSSP and GVDP can't really be directly compared because the OFs have qualitatively
different meanings. However, the general comparison is still useful; AdaSSP performs well with slightly overestimated bounds but scales poorly with overly conservative bounds,
while GVDP performs a bit less well at low overestimation factors but scales much better when the bounds are poorly set. 
More information can be found in Section~\ref{appendix:section:explanation_adassp_comparison}.

\subsection{Replication of~\citet{Car99}}
\label{subsection:replication_of_card}

Inspired by the tests of~\citet{BWSB21}, we attempt to replicate the core analysis of \citet{Car99}
under the constraints of DP. We use CPS ASEC data from 1994 to 1996~\citep{RFFGPSS21} and run OLS 
to estimate the following model:
\[ \log(\text{inc\_wage}) = \beta_0 + \beta_1 \text{educ} + \beta_2 \text{PE} + \beta_3 \text{PE}^2 + \beta_4 \text{PE}^3 + \beta_5 \text{white} + \epsilon, \]
where \emph{inc\_wage} is an individual's total pre-tax wage and salary income, \emph{educ} is 
years of education, \emph{PE} is potential years of work experience, 
and \emph{white} indicated whether or not the individual identifies as white.
The effect of education on income is our question of interest, with the other variables serving as controls.
This allows us to use the full OLS model within the bootstrap, but release and privately estimate only 
the estimated mean/variance of $\beta_1$.

\citet{Car99} runs this model separately for males and females; in Figure~\ref{appendix:fig:card_reanalysis} we report the female results ($n = 95{,}177$) as well as for males 
and females combined ($n = 197{,}756$).
To be consistent with~\citet{BWSB21}, we report results in \emph{approximate DP} with 
$(\epsilon, \delta) = (5, 1/k)$, which translates to $\rho \approx \{ 0.879, 1.06, 1.23 \}$ for $k = \{1000,500,250\}$.
All results are run with $t=5$ CoinPress iterations and an overestimation factor of 100, and we 
run the GVDP estimation algorithm $200$ times to show the long-run performance of the algorithm.

\begin{figure}[h]
    \captionsetup{font=footnotesize,labelfont=footnotesize}
    \begin{subfigure}[t]{.32\textwidth}
        \centering
        \includegraphics[width=\textwidth]{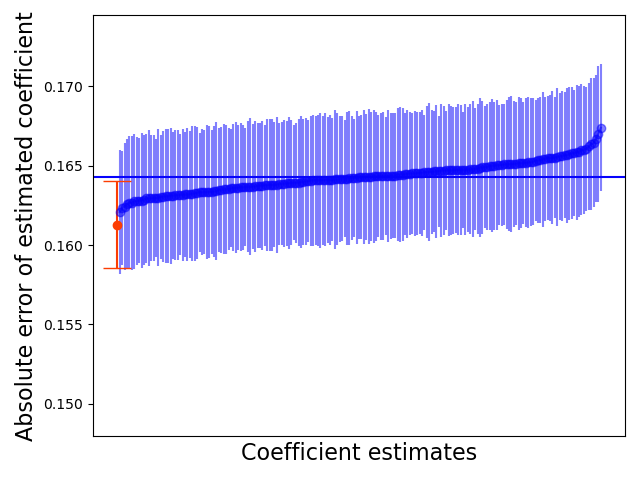}
        \caption{female: $k = 1{,}000$}
    \end{subfigure}
    \begin{subfigure}[t]{.32\textwidth}
        \centering
        \includegraphics[width=\textwidth]{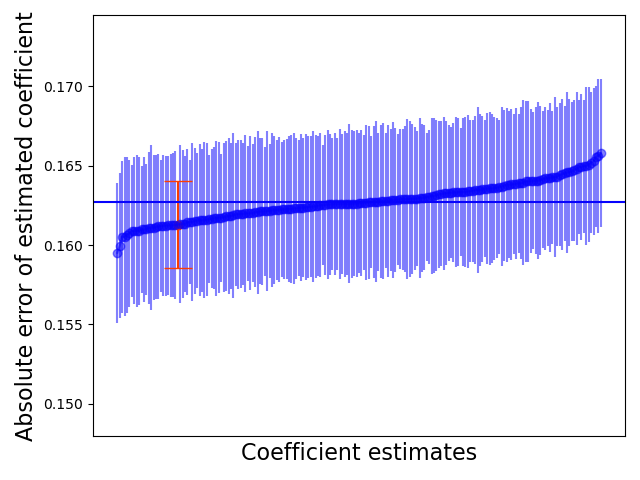}
        \caption{female: $k = 500$}
    \end{subfigure}
    \begin{subfigure}[t]{.32\textwidth}
        \centering
        \includegraphics[width=\textwidth]{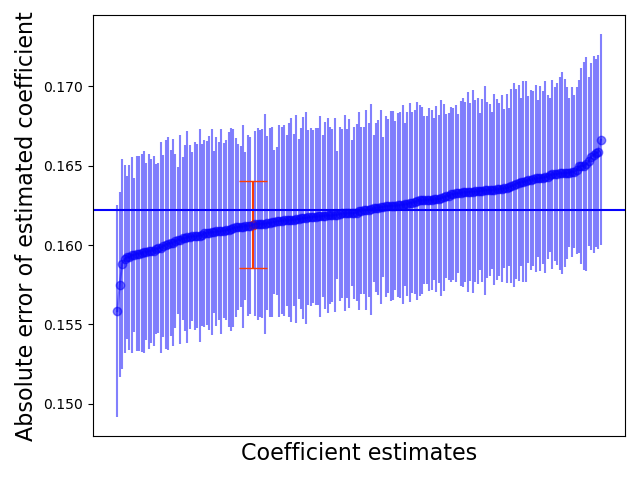}
        \caption{female: $k = 250$}
    \end{subfigure}
    \begin{subfigure}[t]{.32\textwidth}
        \centering
        \includegraphics[width=\textwidth]{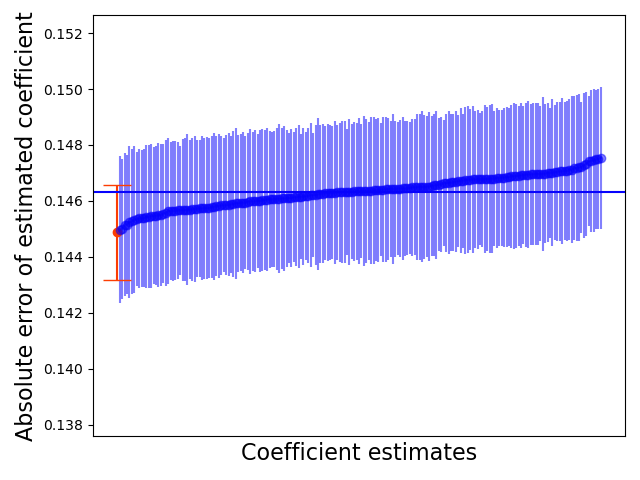}
        \caption{combined: $k = 1{,}000$}
    \end{subfigure}
    \begin{subfigure}[t]{.32\textwidth}
        \centering
        \includegraphics[width=\textwidth]{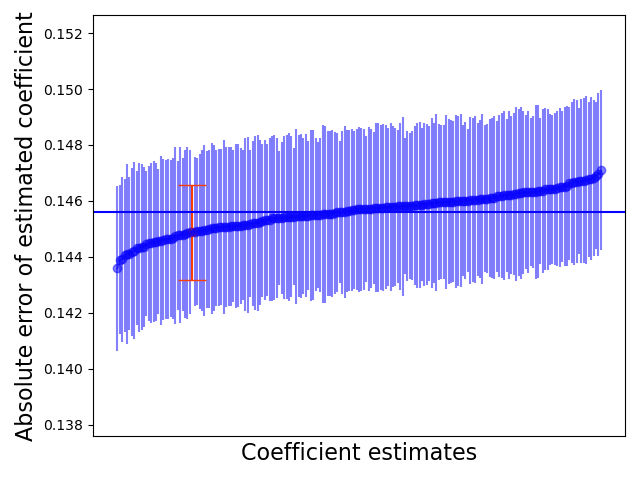}
        \caption{combined: $k = 500$}
    \end{subfigure}
    \begin{subfigure}[t]{.32\textwidth}
        \centering
        \includegraphics[width=\textwidth]{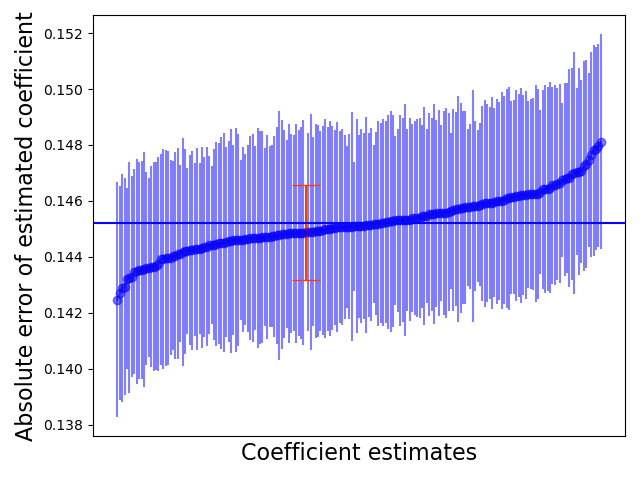}
        \caption{combined: $k = 250$}
    \end{subfigure}
    \caption{Distribution of coefficient estimates and 95\% confidence intervals for  
        females only and both males and females (combined) 
        with $\epsilon = 5$. 
        The dot with a capped error bar represents the non-private estimate and confidence interval.
        The wider bars are the upper/lower bounds on the confidence intervals for the runs of GVDP.
        The horizontal line is the empirical mean of the GVDP estimates.
    } 
    \label{appendix:fig:card_reanalysis}
\end{figure}

We note that our bootstrapped means do not always equal to the non-private mean in expectation, 
so although our algorithm's guarantees with respect to the bootstrapped
distribution are met, they do not imply guarantees relative to the non-private answer as we hope they would.
We see in these plots a clear trade-off. 
At $k = 1{,}000$, our confidence intervals are fairly tight, but are essentially centered around the 
upper end of the non-private confidence interval rather than the the true coefficient estimate.
At $k = 250$ we minimize bias by generating more a representative bootstrap
distribution, but do so at the cost of wider confidence intervals. 

\section{Notes on Assumptions and Analyst Choices}

\subsection{Assumption~\ref{assump:blb_approx}}

Assumption~\ref{assump:blb_approx} essentially has three distinct pieces; we speak to the plausibility of each below.

First, we assume that the sampling distribution of the estimator is a member of a symmetric multivariate location-scale family,
which we require because we want to be able to fully characterize the distribution by its mean and covariance.
Given that $X$ is of a reasonable sample size, we can appeal to the central limit theorem and argue that this assumption ought to hold.
If the analyst cares only about confidence intervals for each element of the parameter,
rather than a joint confidence region, it is sufficient for the marginal sampling distribution of each element
in the parameter vector to belong to a symmetric univariate location-scale family.

Second, we assume that the BLB estimator in unbiased with respect to the estimand of interest in the non-private setting.
The BLB shares many of the statistical properties of the traditional bootstrap,
including asymptotic consistency (as both $n \rightarrow \infty$ and $\frac{n}{k} \rightarrow \infty$),
but also no finite-sample guarantee of unbiasedness. 
As such, this assumption may not hold in practice. However, if the BLB estimator exhibits 
low bias relative to the potential bias induced by poorly chosen clipping bounds, our 
method could still be an effective way to produce private estimates with lower bias than 
existing methods.

Third, we assume that the BLB estimates of the covariance are, with probability 1, a L{\"o}wner upper bound on the true covariance of the sampling distribution. 
This condition on $\hat{\Sigma}^{BLB}$ is onerous (especially in high dimensions) and seems unlikely to hold in general.
In practice however, this condition can be dropped at the cost of a bit of extra fuzziness in the results.
As stated above, we later make claims about our private estimator relative to $\tilde{\Sigma}^{BLB}$,
which under Assumption~\ref{assump:blb_approx} also hold relative to $\hat{\Sigma}$.
This generalization to $\hat{\Sigma}$ is a higher bar than is typically set in applications of the bootstrap, 
where the bootstrap approximation is simply treated as a ``good-enough'' approximation 
of the sampling distribution. Moreover, if we care only about getting confidence intervals, rather 
than a confidence region, we can replace the L{\"o}wner condition with the condition that each element 
of the diagonal of $\hat{\Sigma}^{BLB}$ is at least as large as the corresponding element of 
$\hat{\Sigma}$.

\subsection{Assumption~\ref{assump:user_heavier_tailed}}

Assumption~\ref{assump:user_heavier_tailed} essentially follows from the first part of 
Assumption~\ref{assump:blb_approx} where we assume the sampling distribution is from a symmetric
location-scale family. If we can identify the location-scale family of the sampling distribution 
(again, we often appeal to the central limit theorem and say this is Multivariate Gaussian), 
then this same family trivially satisfies Assumption~\ref{assump:user_heavier_tailed}.

\subsection{Assumptions~\ref{assump:initial_parameter_bounds} and~\ref{assump:initial_covariance_bounds}}

Assumptions~\ref{assump:initial_parameter_bounds} and~\ref{assump:initial_covariance_bounds}
state that the analyst can set bounds on the mean and covariance on the BLB estimates
of both the mean and covariance of the sampling distribution. 
We believe that setting tight bounds would be very difficult in general, often 
more difficult than setting tight bounds on the data (the requirement we're trying to avoid).
However, we suggest that the analyst aim to set very conservative bounds, unless they are very confident 
in their knowledge of the parameters. 
Because we use the CoinPress mean estimation algorithm to iteratively improve the bounds 
the analyst provides, the performance of the algorithm degrades slowly with more conservative bounds; 
e.g.\ see our results from Section~\ref{section:ols_regression_demo} where the analyst's bounds are 
too conservative by a factor of 100 or Section~\ref{subsection:comparison_with_adassp} where we show
performance when the analyst's bounds are too conservative by a factor of 10{,}000.  

\subsection{Choosing $k$}
\label{appendix:subsection:choosing_k}

Recall from our explanation of Algorithm~\ref{alg:gvdp} that $k$ is the number of subsets into 
which we partition our original data, which in turn becomes the number of elements fed into our 
private mean estimation algorithm, Algorithm~\ref{appendix:alg:hd_mean_est}.
This presents a trade-off for the user; when $k$ is large, the sensitivity of our aggregator decreases
(Line~\ref{appendix:ln:hd_mean_step-sens} of Algorithm~\ref{appendix:alg:hd_mean_step}) and thus so does the 
variance of the noise we need to add for privacy. 
On the other hand, we assume that the mean and covariance estimates we get from the BLB 
reasonably approximate the mean and covariance of the true sampling distribution of the parameters, 
which is provably true only as $n \rightarrow \infty$ and 
$\frac{n}{k} \rightarrow \infty$ for Hadamard differentiable estimators~\citep{KTSJ14}.

In the body of the paper, we argued that once $\frac{n}{k}$ is large enough that the BLB estimates have converged, there is 
no use in further increasing the ratio of $n$ to $k$; we are better served by increasing $k$
and reducing the noise needed for privacy. So, the best possible case for an analyst 
is that they choose the largest $k$ such that the BLB estimates, operating 
over subsets of size $\frac{n}{k}$, approximates the true parameters of the sampling distribution.

As a final note, recall from Section~\ref{subsection:algorithm_step_1} that we assume that $n$ is public knowledge 
or has been privately estimated. Thus, we can choose $k$ in a way that depends on $n$ with no extra privacy cost;
if $n$ is public knowledge then there's no dependence on the data at all and if it was privately 
estimated then this falls under the postprocessing property of zCDP.